\documentclass[a4paper]{IEEEtran}

\usepackage{pgfplots}
\pgfplotsset{compat=newest}
\usepackage{tikz}
\usetikzlibrary{arrows,matrix,positioning}
\usepackage[utf8]{inputenc}
\usepackage{caption,subcaption}
\usepackage[ruled,vlined,linesnumbered]{algorithm2e}
\usepackage[hang,flushmargin]{footmisc}
\usepackage{lipsum}
\makeatletter
\newcommand{\algorithmfootnote}[2][\footnotesize]{%
  \let\old@algocf@finish\@algocf@finish%
  \def\@algocf@finish{\old@algocf@finish%
    \leavevmode\rlap{\begin{minipage}{\linewidth}
    {#1#2}
    \end{minipage}}%
  }%
}
\makeatother
\usepackage[english]{babel}
\usepackage{amsmath, amssymb, amsbsy}
\usepackage{mathdots}
\usepackage{mathtools}
\usepackage{cuted}
\usepackage{url}

\usepackage{color}
\usepackage{tabularx}
\usepackage{footnote}

\usepackage{amsthm}
\newtheorem{definition}{Definition}%
\newtheorem{corollary}{Corollary}%
\newtheorem{theorem}{Theorem}%
\newtheorem{lemma}{Lemma}%
\newtheorem{remark}{Remark}

\usepackage{array}
\usepackage{multirow}
\usepackage{arydshln}

\usepackage[capitalise]{cleveref}

\DeclareMathOperator{\rank}{rank}

\DeclareMathOperator{\diag}{diag}

\DeclareMathOperator{\wt}{wt}
\DeclareMathOperator{\supp}{supp}

\newcommand{\F}{\ensuremath{\mathbb{F}}}
\newcommand{\Fq}{\ensuremath{\mathbb{F}_q}}
\newcommand{\Fqm}{\ensuremath{\mathbb{F}_{q^m}}}

\newcommand{\ceil}[1]{\ensuremath{\left\lceil{#1}\right\rceil}}
\newcommand{\floor}[1]{\ensuremath{\left\lfloor{#1}\right\rfloor}}
\newcommand{\supbrac}[2]{\ensuremath{#1^{(#2)}}}

\newcommand{\myspan}[1]{\left\langle #1 \right\rangle}

\newcommand{\code}{\ensuremath{\mathcal{C}}}

\newcommand{\GRS}{\mathsf{GRS}}
\newcommand{\GRSp}{\GRS_{\va,\v}^d}
\newcommand{\GRSallp}{\mathbb{G}_{\va}^d}
\newcommand{\ALTallp}{\mathbb{A}_{\va}^d}

\newcommand{\kopt}{k_q^{\mathsf{opt}}}

\newcommand{\dGoppa}{d_{\mathsf{Goppa}}}

\newcommand{\dcol}{d_\mathsf{col}}

\newcommand{\tmax}{\ensuremath{t_{\max}}}

\newcommand{\Pmisc}{P_{\mathsf{misc}}}
\newcommand{\Psuc}{P_{\mathsf{suc}}}
\newcommand{\Pfail}{P_{\mathsf{fail}}}

\newcommand{\tikznode}[2]{%
	\ifmmode%
	\tikz[remember picture,baseline=(#1.base),inner sep=0pt] \node (#1) {$#2$};
	\else
	\tikz[remember picture,baseline=(#1.base),inner sep=0pt] \node (#1) {#2};%
	\fi}
\tikzset{%
	mybox_block/.style={rectangle,rounded corners,draw=black, thick,text width=1em,minimum height=2em,minimum width=4.75em,text centered},
	[highlight/.style={rectangle,rounded corners,fill=#1!15,draw,fill opacity=0.5,thick,inner sep=0pt},
	highlight/.default=gray],
	plot1/.style = {thick,
		dotted,
		mark=+}
}

\def\ve#1{{\mathchoice{\mbox{\boldmath$\displaystyle #1$}}%
              {\mbox{\boldmath$\textstyle #1$}}%
              {\mbox{\boldmath$\scriptstyle #1$}}%
              {\mbox{\boldmath$\scriptscriptstyle #1$}}}}

\newcommand{\0}{\ensuremath{\ve{0}}}
\newcommand{\1}{\ensuremath{\ve{1}}}

\newcommand{\C}{\ensuremath{\ve{C}}}
\newcommand{\D}{\ensuremath{\ve{D}}}
\newcommand{\E}{\ensuremath{\ve{E}}}
\newcommand{\Fbold}{\ensuremath{\ve{F}}}
\renewcommand{\H}{\ensuremath{\ve{H}}}

\newcommand{\R}{\ensuremath{\ve{R}}}
\renewcommand{\S}{\ensuremath{\ve{S}}}
\newcommand{\T}{\ensuremath{\ve{T}}}
\newcommand{\V}{\ensuremath{\ve{V}}}

\renewcommand{\a}{\ve{a}}

\renewcommand{\c}{\ve{c}}
\newcommand{\e}{\ve{e}}

\renewcommand{\u}{\ve{u}}
\renewcommand{\v}{\ve{v}}

\newcommand{\va}{\ve{\alpha}}
\newcommand{\vLambda}{\ve{\Lambda}}

\newcommand{\cA}{\mathcal{A}}
\newcommand{\cE}{\mathcal{E}}

\newcommand{\cI}{\mathcal{I}}
\newcommand{\cL}{\mathcal{L}}
\newcommand{\cK}{\mathcal{K}}
\newcommand{\cV}{\mathcal{V}}
\newcommand{\cZ}{\mathcal{Z}}
\newcommand{\cM}{\mathcal{M}}

\newcommand{\cS}{\mathcal{S}}
\newcommand{\cU}{\mathcal{U}}

\newcommand{\sfJ}{\mathsf{J}}

\newcommand{\EB}[3]{\mathbb{E}_{#1}^{(#2,#3)}}
\newcommand{\Ebbad}{\mathbb{E}_{w\text{-}\mathsf{bad}}}

\newcommand{\labelRS}{\mathsf{L.RS}}
\newcommand{\labelMain}{\mathsf{L.A}}
\newcommand{\labelSingleton}{\mathsf{L.A1}}
\newcommand{\labelLz}{\mathsf{L.A2}}
\newcommand{\labelLarge}{\mathsf{L.T}}
\newcommand{\labelMisc}{\mathsf{M}}
\newcommand{\labelLower}{\mathsf{U}}
\newcommand{\labelSim}{\mathsf{SIM}}

\IEEEoverridecommandlockouts

\begin{document}

\title{Decoding of Interleaved Alternant Codes}
\author{\IEEEauthorblockN{Lukas Holzbaur, \IEEEmembership{Student Member, IEEE}, Hedongliang Liu,~\IEEEmembership{Student Member, IEEE}, \\Alessandro Neri, Sven Puchinger, \IEEEmembership{Member, IEEE}, Johan Rosenkilde, \\Vladimir Sidorenko,~\IEEEmembership{Member, IEEE}, Antonia Wachter-Zeh,~\IEEEmembership{Senior Member, IEEE}} %
  \thanks{Parts of this paper have been presented at the \emph{2020 IEEE Information Theory Workshop (ITW)} \cite{holzbaur202sucessProbability}.\newline
    The work of L.~Holzbaur and A.~Wachter-Zeh has been supported by the German Research Foundation (Deutsche Forschungsgemeinschaft, DFG) under grant no.~WA3907/1-1 and by the European Research Council (ERC) under the European Union’s Horizon 2020 research and innovation programme (grant agreement no.~801434).
  	The work of H.~Liu has been supported by the German  Research  Foundation (DFG) with a German Israeli Project Cooperation (DIP) under grants no.~PE2398/1-1, KR3517/9-1.
  	S.~Puchinger has received funding from the European Union’s Horizon 2020 research and innovation programme under the Marie Skłodowska-Curie grant agreement no.~713683. 
  	The work of A.~Neri has been supported by the Swiss National Science Foundation under the Early Postdoc.Mobility grant no.~187711.  \newline 
    L.~Holzbaur, H.~Liu, S.~Puchinger, V.~Sidorenko, and A.~Wachter-Zeh are with the Institute for Communications Engineering, Technical University of Munich (TUM), Germany. A.~Neri is with the Max Planck Institute for Mathematics in the Sciences, Leipzig. J.~Rosenkilde is with GitHub, Inc.\newline %
    Emails: lukas.holzbaur@tum.de, lia.liu@tum.de, alessandro.neri@mis.mpg.de, svepu@dtu.dk, jsrn@jsrn.dk, vladimir.sidorenko@tum.de, antonia.wachter-zeh@tum.de}}

\maketitle

\begin{abstract}
  Interleaved Reed--Solomon codes admit efficient decoding algorithms which correct burst errors far beyond half the minimum distance in the random errors regime, e.g.,~by computing a common solution to the Key Equation for each Reed--Solomon code, as described by Schmidt et~al.
  If this decoder does not succeed, it may either \emph{fail} to return a codeword or \emph{miscorrect} to an incorrect codeword, and good upper bounds on the fraction of error matrices for which these events occur are known.

  The decoding algorithm immediately applies to interleaved alternant codes as well, i.e.,~the subfield subcodes of interleaved Reed--Solomon codes, but the fraction of decodable error matrices differs, since the error is now restricted to a subfield.
  In this paper, we present new general lower and upper bounds on the fraction of error matrices decodable by Schmidt et al.'s decoding algorithm, thereby making it the only decoding algorithm for interleaved alternant codes for which such bounds are known. %

\end{abstract}


\section{Introduction}
\label{sec:introduction}

A codeword of an $\ell$-interleaved code can be seen as $\ell$ codewords of possibly different codes stacked above each other, i.e., an $\ell$-interleaved code is the direct sum of $\ell$ codes of the same length $n$, and its codewords may be represented as $\ell \times n$ matrices over the base field $\F$.
A common error model for these codes are \emph{burst errors}, where we assume an error corrupts an entire column, and as distance metric we count the number of non-zero columns of such an $\ell \times n$ matrix.
When considering the words as vectors in $\mathbb{L}^n$, where $\mathbb{L}$ is an extension of $\F$ of degree $\ell$, this is corresponds to the Hamming distance.
In this work, we analyse the fraction of decodable error patterns for a given error weight under this metric. If the error is uniformly distributed over the set of all such error matrices, this is equivalent to the probability of successful decoding.

To decode an interleaved code, we may simply decode each constituent codeword, i.e.~consider each row of the $\ell \times n$ matrix independently.
However, for a variety of algebraic interleaved codes, it is possible to correct a larger fraction of errors by adopting a collaborative approach.
For this reason, interleaved codes have many applications in which burst errors occur naturally or artificially, for instance
replicated file disagreement location \cite{metzner1990general},
correcting burst errors in data-storage applications \cite{krachkovsky1997decoding,holzbaur2019error},
outer codes in concatenated codes \cite{metzner1990general,krachkovsky1998decoding,haslach1999decoding,justesen2004decoding,schmidt2005interleaved,schmidt2009collaborative},
ALOHA-like random-access schemes \cite{haslach1999decoding},
decoding non-interleaved codes beyond half-the-minimum distance by power decoding \cite{schmidt2010syndrome,kampf2014bounds,rosenkilde2018power,puchinger2019improved},
and code-based cryptography \cite{elleuch2018interleaved,holzbaur2019decoding}.

Generalized Reed--Solomon (GRS) codes are among the most-studied classes of constituent codes for interleaved codes.
There are several decoders for interleaved GRS codes \cite{krachkovsky1997decoding,bleichenbacher2003decoding,brown2004probabilistic,schmidt2009collaborative,nielsen2013generalised,yu2018simultaneous} that decode up to $\tfrac{\ell}{\ell+1}(n-\bar{k})$ errors, where $\ell$ is the interleaving order (number of constituent codes) and $\bar{k}$ is the mean dimension of the constituent codes.
All of these decoders fail for \emph{some} error patterns of weight larger than the unique decoding radius of the constituent code of lowest distance (which is also the distance of the interleaved code).
For errors of a given weight, the fraction of errors leading to a unsuccessful decoding is roughly $q^{-m}$ at the maximal decoding radius (where $q^m$ is the field size of the GRS code), and decreases exponentially in the difference of the maximal decoding radius and the actual error weight.

There are also various other decoding algorithms for interleaved GRS codes that decode beyond the radius $\tfrac{\ell}{\ell+1}(n-\bar{k})$, and even beyond the Johnson radius: \cite{coppersmith2003reconstructing,parvaresh2004multivariate,parvaresh2007algebraic,schmidt2007enhancing,cohn2013approximate,wachterzeh2014decoding,puchinger2017irs}.
For some of these decoders, simulation results suggest that these decoders can successfully decode a large fraction of error matrices of weight up to the claimed maximal radius, and in some very special cases, it is possible to derive bounds on this fraction.
However, in general, only little is known about the fraction of decodable errors for these decoders, which are therefore not considered in this work. Other code classes that have been considered as constituent codes of interleaved codes are one-point Hermitian codes \cite{kampf2014bounds,puchinger2019improved} and, more generally, algebraic-geometry codes~\cite{brown2005improved}.

For interleaved decoders of high order, i.e., where $\ell$ is larger than the weight of the error, a simple linear-algebraic decoder was proposed in \cite{metzner1990general}.
Unlike all decoders mentioned above, this decoder works with interleaved codes obtained from an arbitrary linear constituent code and guarantees to correct any error of weight up to $d-2$ that has full rank, where $d$ is the minimum distance of the constituent code.
It was rediscovered in \cite{haslach1999decoding} and generalized in \cite{haslach2000efficient,roth2014coding}.

An \emph{alternant code} is a subfield subcode of a GRS code: the set of codewords whose entries are all contained in a fixed subfield of the base field of the GRS code. This code family contains some of the best-known and most-often used algebraic codes over small fields, including the Bose–Ray-Chaudhuri–Hocquenghem (BCH) and Goppa codes.
In principle, alternant codes can be used as constituent codes in any of the above mentioned applications of interleaved codes.
We see several concrete reasons to specifically consider alternant codes: %
\begin{itemize}
\item Alternant codes (especially BCH codes) are some of the most-often used algebraic codes in practice, including for storage and communications. Any system that already uses these codes and is prone to burst errors may be retroactively upgraded to enable a larger error-correction capability.
For instance, in NOR and NAND flash memory, Hamming and BCH codes are considered as the standard error correction approach (cf. \cite{ECCFlash_Book,LowPowerBCHNAND,VLSI-BCH-NAND}).
Traditionally, Hamming codes are used in single-level flash memories to correct single errors as they have a simple decoding algorithm and use only a small circuit area. For multi-level flash memories however, single-error correction is not sufficient and BCH codes with larger distance are employed. In \cite{StrongBCH-NAND_SIPS2006}, the scenario of more than four levels (i.e., storing more than two bits per flash memory cell) was investigated and it was shown that BCH codes of larger correction capability are needed.
To address the fact that errors in flash memories might occur over whole bit or word lines, in \cite{FlashProductBCH} product codes with BCH codes were used. This motivates the use of \emph{interleaved} alternant and in particular interleaved BCH codes.

\item In applications where the cost of \emph{encoding} is dominant (e.g., in storage systems where writing occurs more often than reading an erroneous codeword), encoding in a subfield reduces the complexity. Hence, it might be advantageous to use alternant codes instead of GRS codes in some of the above mentioned applications of interleaved codes. Note that \emph{decoding} is usually done in the field of the corresponding GRS code, so the reduction in complexity is less significant.

\item In some applications, such as code-based cryptography, GRS and algebraic-geometry codes cannot be used due to their vast structure, which can be turned into structural attacks on the cryptosystem.
However, their subfield subcodes are in many cases unbroken (cf.~\cite[Conclusion]{couvreur2017cryptanalysis} and \cite[Section~7.5.3]{couvreur2020algebraic}).
In particular, the codes proposed in McEliece's original paper \cite{McE78}, binary Goppa codes, have withstood efficient attacks for more than $40$ years.
In a McEliece-type system, the ciphertext is the sum of a codeword of a public code and a randomly chosen ``error'' which hides the codeword from the attacker.
If we encrypt multiple codewords in parallel, we may consider them as an interleaved code and align the errors in bursts of larger weight.
This approach has the potential to increase the designed security parameter, or in turn reduce the key size, and was first studied in \cite{elleuch2018interleaved,holzbaur2019decoding}.
This comes at the cost of a (hopefully very small) probability of unsuccessful decryption/decoding, which corresponds to the probability of unsuccessful decoding of the interleaved decoder.
\end{itemize}

Interleaved alternant codes can be decoded by the decoders of interleaved GRS codes. However, the set of all errors of a given weight differs for interleaved alternant codes, as it only contains matrices over the subfield corresponding to the alternant code, not the field of the GRS code. Therefore, the bounds on the fraction of decodable error matrices for the decoding of interleaved GRS codes do not apply to interleaved alternant codes. Aside from a theoretical interest, it is crucial for all of the above mentioned applications to estimate this fraction, or, equivalently, the probability of successful decoding for errors drawn uniformly at random from this set.

In this paper, we derive lower bounds on probability of success for decoding interleaved alternant codes with the decoder from \cite{FengTzeng1991ColaDec,schmidt2009collaborative} for uniformly distributed errors of a given weight. Further, for comparison, we also derive upper bounds on the probability of successful decoding.
To the best of our knowledge, this is the first work that studies the success probability of decoding interleaved alternant codes for general parameters.

\subsection{Overview \& Main Results}

The remainder of the paper is organized as follows: In~\cref{sec:preliminaries} we define the notation used throughout the paper. We shortly recap the syndrome based interleaved decoder from \cite{FengTzeng1991ColaDec,schmidt2009collaborative} and formally define the event of a decoding failure and a miscorrection.
We derive a necessary and sufficient condition for the decoder to succeed, which simplifies the subsequent analyses.
\cref{sec:TechnicalPreliminaries} establishes some technical preliminary results which are then used in~\cref{sec:upperBounds,sec:lowerBound} for the derivation of our main results:
\begin{itemize}
\item \cref{thm:failureProbNewBoundGeneral} provides a framework for lower bounding the probability of decoding success for interleaved alternant codes with the decoder of \cite{FengTzeng1991ColaDec,schmidt2009collaborative}, by relating it to properties of the set of all alternant codes obtained from the generalization of specific RS codes.
  Based on this framework, \cref{thm:failureProbNewBound} presents a lower bound on the probability of successful decoding by applying the technical results established in \cref{sec:TechnicalPreliminaries}.
\item \cref{thm:LargeEll} gives an alternative lower bound based on the ideas from \cite{metzner1990general,roth2014coding}, which improves upon the bound of \cref{thm:failureProbNewBound} for some parameters. In particular, for large interleaving order $\ell$, this bound provides non-trivial results even when the number of errors is close to the maximum decoding radius of the corresponding GRS code.
\item \cref{thm:lowerBound} gives an upper bound on the probability of success for decoding interleaved alternant codes with the considered decoder. This result allows us to evaluate the performance of the lower bounds presented in \cref{thm:failureProbNewBound,thm:LargeEll}.
\end{itemize}
In~\cref{sec:numericalResults} we present numerical evaluations of our bounds for different code parameters and discuss their implications.
Finally, we conclude the paper and discuss some open problems in~\cref{sec:conclusion}.%

\section{Preliminaries}
\label{sec:preliminaries}

\subsection{Notation}
We denote by $[a,b]$ the set of integers $\{i \mid a\leq i \leq b\}$ and if $a=1$, we omit it from our notation and write $[b]$.

A finite field of size $q$ is denoted by $\Fq$ and $\Fq^{\star}\coloneqq \Fq\setminus\{0\}$. Vectors are denoted by bold lower-case letters and matrices by bold capital letters.

Given a vector $\a\in \Fq^b$, we denote by $\diag(\a)$ the diagonal matrix with entries of $\a$ in the main diagonal. For a set of integers $\cL \subseteq [b]$, we denote by $\a|_{\cL}$ the restriction of $\a$ to the entries indexed by $\cL$. Denote by $\wt(\a)$ the Hamming weight of the vector $\a$ and by $|\a|$ the length of $\a$.

Given a matrix $\E\in \Fq^{a\times b}$, we denote by $\E|_{\cL}$ the restriction of $\E$ to the columns indexed by $\cL$, by $\E_{i,:}$ its $i$-th row, and by $\supp(\E)$ the set of indices of the non-zero columns of $\E$.
Denote by $\EB{q}{a}{b}$ the set of matrices $\E \in \Fq^{a\times b}$ with at least one non-zero element in each column.

We write $[n,k,d]_q$ to denote a linear code $\code \in \Fq^n$ of dimension $k$ and minimum distance at least $d$.
The cardinality of a set (code) $\cS=\{s_1,s_2,\ldots \}$ is denoted by $|\cS|$. For a multiset $\cS = \{\{s_1,\ldots, s_1,s_2,\ldots, s_2, \ldots\}\}$ we denote by $\delta^{s_i}_\cS$ the multiplicity of $s_i$ in $\cS$. For a linear subspace $\cV \subseteq \Fq^n$ we denote its dimension by $\dim_q(\cV)$.
For a random variable $X$ with uniform probability distribution over a set $\cS$, we write $X \sim \cS$.

\subsection{Generalized Reed-Solomon Codes and their Subfield Subcodes}

We begin by formally defining the class of generalized RS codes.
\begin{definition}[Generalized Reed-Solomon Codes]\label{def:GRScodes}
  For positive integers $d$ and $n$, let $\va \in (\Fqm^\star)^n$ be a vector of distinct \emph{code locators} and $\v \in (\Fqm^\star)^n$ be a vector of \emph{column multipliers}.
  We define a \emph{generalized Reed-Solomon} (GRS) code $\GRSp$ as
    \begin{align*}
      \GRSp &= \{\c \in\Fqm^n \mid \H\cdot \diag(\v)\cdot \c=\0\}\ ,
    \end{align*}
    with
  \begin{align*}
    \H =
    \begin{pmatrix}
      1 & 1 & \dots & 1\\
      \alpha_1 & \alpha_2 & \dots & \alpha_n \\
      \vdots & \vdots &  & \vdots\\
      \alpha_1^{d-2} & \alpha_2^{d-2} & \dots & \alpha_n^{d-2}
    \end{pmatrix} \ \in \ \Fqm^{(d-1)\times n}\ .
  \end{align*}
  Denote by $\mathbb{G}_{\va}^d$ the multi-set
  \begin{equation*}
    \mathbb{G}_{\va}^d = \{\{ \GRSp \ | \ \v \in (\Fqm^\star)^n \}\} \ .
  \end{equation*}
\end{definition}

Note that the most general definitions of GRS codes allow for the $\alpha_i = 0$ to be element of $\va$, but for consistency with \cite{schmidt2009collaborative} and as this complicates the decoding process, we restrict ourselves to $\alpha_i \neq 0$ here.
GRS codes are well-known to be so-called \emph{Maximum Distance Separable} (MDS) codes, i.e., they achieve $d=n-k+1$, where $k$ is the dimension of the code.

The weight enumerator $A_{w}^\code$, i.e., the number of codewords of Hamming weight $w$ in a code $\code$, is completely determined by the code parameters length and distance/dimension if $\code$ is MDS.
\begin{theorem}[MDS Code Weight Enumerator{\cite[Ch.~11, Theorem~6]{macwilliams1977theory}}]\label{thm:MDSWeightEnumerator}
  Let $\code$ be an $[n,k,d]_{q^m}$ MDS code. The $w$-th weight enumerator $A_w^{\mathsf{MDS}}$ of $\code$ is $A_0^{\mathsf{MDS}}=1$ and
  \begin{align*}
    A_w^{\mathsf{MDS}} &\coloneqq |\{ \c \ | \ \wt(\c) = w, \c \in \code \} | \\
    &= \binom{n}{w} \sum_{j=0}^{w-d} (-1)^j \binom{w}{j} (q^{m(w-d+1-j)}-1) , \ \ w\neq 0 \ .
  \end{align*}
\end{theorem}

By design, GRS codes must be defined over fields $\Fqm$ with $q^m-1\geq n$ (or $q^m\geq n$ if $\alpha_i=0$ is allowed as a code locator). In many applications it is desirable to work with codes of smaller field size, which can be obtained, e.g., by taking subcodes of codes defined over larger fields.

\begin{definition}[Subfield Subcode]\label{def:subfieldSubcode}
  Let $\code$ be an $[n,k,d]_{q^m}$ code. We define the $\F_q$-subfield subcode of $\code$ as
  \begin{equation*}
    \code \cap \Fq^n = \{\c \ | \ \c \in \code, c_i \in \Fq \ \forall \ i \in [n]\} \ .
  \end{equation*}
  Equivalently, let $\H \in \F_{q^m}^{(n-k) \times n}$ be a parity check matrix of $\code$. Then $\code \cap \Fq^n$ is given by the $\F_q$ kernel of $\H$, i.e.,
  \begin{align*}
     \code \cap \Fq^n = \{\c \ | \ \H \cdot \c = \0 , \c \in \F_q^n\} \ .
  \end{align*}
\end{definition}

In this work we consider codes from the class of subfield subcodes of GRS codes.

\begin{definition}[Alternant Code~{\cite[Ch.~12.2]{macwilliams1977theory}}]\label{def:AlternantCode}
  The subfield subcode of a GRS code is referred to as an \textbf{alternant code}. For a fixed set of code locators $\va$ as in \cref{def:GRScodes} and \emph{designed} distance $d$, we define the multi-set of alternant codes as
  \begin{equation*}
    \mathbb{A}_{\va}^d = \{\{ \code \cap \Fq^n \ | \ \code \in \mathbb{G}_{\va}^d \}\} \ .
  \end{equation*}
\end{definition}
We define $\mathbb{A}_{\va}^d$ as a multiset, as the multiplicities will be important in the following. One further advantage is that for a given code length $n=|\va|$ we know its cardinality to be
\begin{align}
  |\mathbb{A}_{\va}^d| = (q^m-1)^{n} \ . \label{eq:cardinalityAall}
\end{align}

 We give some general well-known bounds on the dimension of the $\Fq$-subcode of an $\Fqm$-linear code $\code$ in terms of the parameters of~$\code$.

\begin{lemma}\label{lem:AlternantDimBounds}
  Let $\code$ be an $[n,k,d]_{q^m}$ code. Then
\begin{equation*}
    \max\{n-m(n-k),0\} \leq \dim_q(\code \cap \Fq^n) \leq \min\{k,k_q^{\mathsf{opt.}}(n,d)\} \ ,
\end{equation*}
 where $k_q^{\mathsf{opt.}}(n,d)$ is an upper bound on the dimension of a $q$-ary linear code of length $n$ and minimum distance $d$.
\end{lemma}
\begin{proof}
  The lower bound of $0$ is trivial. The lower bound of $n-m(n-k)$ follows from expanding the $n-k$ rows of any parity-check matrix of $\code$ over some basis of $\Fqm$ over $\Fq$. The resulting $m(n-k) \times n$ matrix is a parity check matrix of the $\F_q$-subcode of $\code \cap \Fq^n$ and the bound follows.

  The upper bound of $k_q^{\mathsf{opt.}}(n,d)$ follows from the fact that the distance of the code $\code \cap \Fq^n$ is at least that of $\code$. Finally,  if $\ell$ elements in $\Fq^n$  are $\Fq$-linearly indpendent, then they are also $\Fqm$-linearly independent for every extension field $\Fqm$ of $\Fq$. Therefore, $\dim_q(\code\cap \Fq^n)\leq k$.
\end{proof}

\begin{remark}[Dimension vs. Distance of binary BCH and Wild Goppa Codes] \label{rem:BCHgoppaCodes}
 Wild Goppa codes~\cite{WildGC76,wild1988Wirtz}, which include \emph{binary square-free Goppa codes}~\cite{GVD70,GVD71,Ber73}, are a subclass of Goppa Codes. Along with BCH codes~\cite{BCH1959H,BCH1960BC}, Goppa codes are the best known class of alternant codes, due to their good distance properties.
 Consider the binary BCH and $q$-ary wild Goppa codes that are subfield subcodes of a GRS code in~$\mathbb{G}^d_{\va}$ for some $\va$ and $d$.

 For binary BCH codes, it is well-known (cf.~\cite[Ch.~7]{macwilliams1977theory}) that their dimension is $k_\mathsf{BCH} \geq n-m\frac{n-k}{2}$, for length $n\coloneqq |\va|$ and dimension $k\coloneqq n-d+1$ of the corresponding GRS code. Therefore, the dimension of binary BCH codes exceeds the generic lower bound of \cref{lem:AlternantDimBounds}.

 Wild ($q$-ary) Goppa codes on the other hand are often considered as alternant codes of $\mathbb{A}_{\va}^{d}$, but with an increased minimum distance $\dGoppa \approx \frac{q}{q-1} d$. However, the bounds presented in this paper depend only on the properties of the corresponding GRS and, in particular, its distance $d$, but not on the \emph{actual} dimension or distance of the considered alternant code itself. Therefore, instead of viewing wild Goppa codes as alternant codes in $\mathbb{A}^{d}_{\va}$ with increased distance, it is convenient to view them as alternant codes of $\mathbb{A}_{\va}^{\dGoppa}$ with a larger \emph{dimension} than guaranteed by the lower bound in \cref{lem:AlternantDimBounds}. This is possible as the improvements of wild Goppa codes compared to alternant codes in general can be shown by proving an equivalence between the Goppa codes obtained from different Goppa polynomials (cf.~\cite{WildGC76}, \cite[Theorem 4.1]{WildBern2011}), which directly implies that $\code_{\mathsf{Goppa}} \in \ALTallp \cap \mathbb{A}^{\dGoppa}_{\va}$ for $\dGoppa >d$.
 Clearly, the ``good'' distance follows immediately from the code being in $\mathbb{A}_{\va}^{\dGoppa}$, while the dimension can be shown to be large by applying the lower bound of \cref{lem:AlternantDimBounds} corresponding to $\ALTallp$.

\end{remark}
For GRS codes it is known \cite{Delsep1975} that for a fixed set of code locators $\va$, it holds that $\GRS_{\va,\v}^d = \GRS_{\va, \u}^d$ if and only if $\v$ is an $\Fqm$-multiple of $\u$, i.e., any code $\code \in \GRSallp$ occurs with multiplicity exactly $\delta^{\code}_{\GRSallp} = q^m-1$ in $\GRSallp$. This gives a lower bound on the multiplicity of alternant codes as
\begin{align}
  \delta^{\cA}_{\ALTallp} \geq q^m-1 \ \forall \ \cA \in \ALTallp  \ . \label{eq:multiplicityAlternant}
\end{align}

\subsection{Decoding of Interleaved Alternant Codes}\label{sec:decoding_details}

We formally introduce the concept of interleaved codes and, for completeness, briefly recap the decoding algorithm of \cite{FengTzeng1991ColaDec,schmidt2009collaborative}.

\begin{definition}[Interleaved Codes]
  The $\ell$-interleaved code $\cI\code^{(\ell)}$ with constituent
  code $\code$ is defined as
  \begin{align*}
    \cI\code^{(\ell)}\coloneqq \left\{
    \begin{pmatrix}
      \c^{(1)}\\
      \vdots\\
      \c^{(\ell)}
    \end{pmatrix}\ | \ \c^{(i)}\in \code,\ i\in[\ell]
\right\}\ .
  \end{align*}
  The parameter $\ell$ is referred to as the \emph{interleaving order} of the interleaved code.
\end{definition}

Let $\cI\code^{(\ell)}$ be an $\ell$-interleaved alternant code with $\code\in\mathbb{A}^d_{\va}$ and $\H$ be the parity-check matrix of the corresponding $\GRSp \in \GRSallp $ of $\code$.
Consider a channel where burst errors of column weight $t$ occur. We transmit a codeword $\C \in \cI\code^{(\ell)}$ of an $\ell$-interleaved alternant code. The received word is given by
\[
  \R = \C + \widetilde{\E} \in \Fq^{\ell \times n} \ ,
\]
where each row of $\C \in \Fq^{\ell \times n}$ is a codeword of $\code$ and $\widetilde{\E} \in \Fq^{\ell \times n}$ has exactly $t$ non-zero columns. Since $\code\subset \GRS$, the received word $\R$ can be decoded by a \emph{syndrome-based collaborative decoding algorithm} for interleaved alternant codes. Such algorithms, to name a few, can be found in~\cite{FengTzeng1991ColaDec} for BCH codes and \cite{krachkovsky1997decoding,bleichenbacher2003decoding,schmidt2009collaborative} for interleaved RS code.
We briefly recapitulate the decoding method below and summarize a naive version of~\cite[Algorithm 2]{schmidt2009collaborative} in~\cref{algo:SyndromeDecoder}.

From the received matrix $\R$, we are able to calculate the syndromes of each row of $\R$ by:
\begin{align}\label{eq:Syndromes}
  \begin{pmatrix}
    \S_{1,:}\\
    \S_{2,:}\\
    \vdots\\
    \S_{\ell,:}
  \end{pmatrix}
  \ =\ \R \cdot \H^\top =  \widetilde{\E} \cdot \H^\top \ ,%
\end{align}
where $\S_{i,:}=(S_{i,1},\dots,S_{i,d-1})\in\Fqm^{d-1},$ for each  $i\in[\ell]$. %

Define the \emph{error locator polynomial} by\footnote{Since $\alpha_i \neq 0$ by \cref{def:GRScodes}, the error locator polynomial is well-defined.}
\begin{align}\label{eq:ELP}
  \Lambda(x)\coloneqq\prod_{i=1}^{t}(1-\alpha^{-1}_{j_i}x)=1+\Lambda_1 x+\dots+\Lambda_t x^t\ ,
\end{align}
where the $t$ roots $\{\alpha_{j_1},\dots,\alpha_{j_t}\}$ of $\Lambda(x)$ are the code locators corresponding to the error positions. The vector of coefficients of $\Lambda(x)$, denoted by $\vLambda$, fulfill the following linear equations, (cf.~\cite{DecRS1960Peterson})

\begin{align}\label{eq:STmatrix}
  &\underbrace{\begin{pmatrix}
      S_{i,1} & S_{i,2} & \dots & S_{i,t}\\
      S_{i,2} & S_{i,3} & \dots & S_{i,t+1}\\
      \vdots & \vdots & & \vdots \\
      S_{i,d-1-t} & S_{i,d-1-t+1} & \dots &S_{i,d-2}
  \end{pmatrix}}_{\S^{(i)}(t)}%
  \begin{pmatrix}
    \Lambda_t\\
    \Lambda_{t-1}\\
    \vdots\\
    \Lambda_1
  \end{pmatrix} \nonumber \\
             &\hspace{4cm} =
                 \underbrace{\begin{pmatrix}
                   -S_{i,t+1}\\
                   -S_{i,t+2}\\
                   \vdots\\
                   -S_{i,d-1}
                 \end{pmatrix}}_{\supbrac{\T}{i}(t)}
\ ,\ \forall i\in[\ell]\ .
\end{align}
Thus, determining the error positions in $\widetilde{\E}$ is equivalent to solving the following linear system of equations $\mathfrak{S}(t)$ for $t$ unknowns, %
\begin{align}\label{eq:KeyEquationLSE}
  \underbrace{\begin{pmatrix}
    \supbrac{\S}{1}(t)\\
    \supbrac{\S}{2}(t)\\
    \vdots\\
    \supbrac{\S}{\ell}(t)
  \end{pmatrix}}_{\S(t)}%
\underbrace{  \begin{pmatrix}
    \Lambda_t\\
    \Lambda_{t-1}\\
    \vdots\\
    \Lambda_1
  \end{pmatrix}
}_{\vLambda}
  &=
                 \underbrace{\begin{pmatrix}
                   \supbrac{\T}{1}(t)\\
                   \supbrac{\T}{2}(t)\\
                   \vdots\\
                   \supbrac{\T}{\ell}(t)
                 \end{pmatrix}}_{\T(t)}\ .
\end{align}

After determining $\vLambda$ from~\cref{eq:KeyEquationLSE}, we may use a standard method for error evaluation such as \emph{Forney's algorithm}~\cite{Forney1965} (cf.~\cite[Section~6.6]{roth2006}) to calculate the error values $\hat{\E}$. %
Then, by subtracting the calculated error $\hat{\E}$ from $\R$, we obtain the estimated codeword $\hat{\C}=\R-\hat{\E}$.

\begin{algorithm}[htb!]
  \caption{Syndrome-based Collaborative Decoding Algorithm}\label{algo:SyndromeDecoder}
  \SetAlgoLined
  \DontPrintSemicolon
  \KwIn{received word $\R$}
  \KwOut{$\hat{\C}$ or \texttt{decoding failure}}
  Calculate the syndromes $\S_{i,:},\forall i\in[\ell]$ \tcp*[r]{See~\cref{eq:Syndromes}} 
  \lIf{$\S_{i,:} = \0$ for all $i$}{\Return{$\hat{\C}=\R$}\label{step:zeroSyndromes}}
    Find minimal $t^\star$ for which $\S(t^\star) \cdot \vLambda^\star = \T(t^\star)$ has a solution and find a solution $\vLambda^\star$ \label{step:solveLSE} \tcp*[r]{See~\cref{eq:KeyEquationLSE}}
    \textbf{if} the solution $\vLambda^\star$ is not unique \textbf{then} output \texttt{decoding failure} and \textbf{stop} \label{step:uniqueSolution}\;
    \uIf{$\Lambda^\star(x)$ \emph{has} $t^\star$ distinct \emph{roots in} $\Fqm$}{
      Evaluate the errors $\hat{\E}$ by Forney's algorithm~\cite{Forney1965}\cite[Section~6.6]{roth2006}\;
      Calculate $\hat{\C}=\R-\hat{\E}$ \label{step:codewordEstimate}}
    \Else{Output \texttt{decoding failure}\label{step:distinctRoots}}
  \SetAlgoLined
\end{algorithm}

For a channel adding errors with some distribution, the collaborative decoding algorithm given in \cref{algo:SyndromeDecoder} may yield three different results:
\begin{itemize}
\item The algorithm returns the correct result, i.e., $\hat{\C}=\C$, with \emph{success probability} $\Psuc$.
\item The algorithm returns an erroneous result, i.e., $\hat{\C}\neq \C$, with \emph{miscorrection probability} $\Pmisc$.
\item The algorithm returns a \texttt{decoding failure}, with \emph{failure probability}~$\Pfail$.
\end{itemize}

\begin{remark}[Practical Implementations]
  \cref{algo:SyndromeDecoder} is a naive approach.
  It is mainly meant for the proof of the successful probability, instead of for an efficient implementation.

  For practical implementations, one can use some fast algorithm for~\cref{step:solveLSE}, for instance, 1)~\cite[Algorithm 3]{SidSchmidtPPI} with the complexity of $O(\ell d^2)$ operations in $\Fqm$, 2) the currently fastest algorithm~\cite{RosenkildeStorjohann21} with complexity $O^{\sim}(\ell^{\omega-1} d)$ where $O^{\sim}$ omits the $\log$-factors in $d$ and $\omega$ is the matrix multiplication exponent, for which the best algorithm allow $\omega<2.38$~\cite{MatMul1990,FastMatMul2014}. %
\end{remark}

\cref{algo:SyndromeDecoder} yields a \emph{bounded distance} decoder which can decode beyond half of the minimum distance $\floor{\frac{d-1}{2}}$ with high probability.
Clearly, the solution $\vLambda^\star$ cannot be unique if the number of equations in~\cref{eq:KeyEquationLSE} is less than the number of unknowns. Thus, we derive the maximum decoding radius of~\cref{algo:SyndromeDecoder} in the following theorem.

\begin{theorem}[Maximum Decoding Radius~{\cite[Theorem~3]{schmidt2009collaborative}}]\label{thm:maxDecodingRadius}
  Let $\cI\code^{(\ell)}$ be an $\ell$-interleaved alternant code with $\code \in \mathbb{A}^d_{\va}$. For a received word $\R=\C+\widetilde{\E}$, where $\C\in\cI\code^{(\ell)}$ and the error $\widetilde{\E}$ has $t$ nonzero columns, \cref{algo:SyndromeDecoder} may only succeed, i.e., return $\hat \C = \C$, if $t$ satisfies
\begin{align}
  t\leq \tmax &\coloneqq \frac{\ell}{\ell+1}(d-1) \ . \label{eq:tmaxGRS}
\end{align}

\end{theorem}
\begin{proof}
  There are $t$ unknowns and $\ell(d-1-t)$ equations in the linear system of equations~\cref{eq:KeyEquationLSE}, which cannot give a unique solution for $t$ unknowns $\Lambda_1,\dots,\Lambda_t$ if the number of unknowns is larger than the number of equations, i.e., we may only obtain a unique solution from~\cref{eq:KeyEquationLSE} if
\begin{align*}
  t& \leq \ell(d-1-t)\ .
\end{align*}
The statement is proved by solving the inequality for $t$.
\end{proof}

By the nature of a bounded distance decoder, where correction spheres inevitably overlap for some error patterns of weight $t$ larger than half minimum distance $\floor{\frac{d-1}{2}}$, \cref{algo:SyndromeDecoder} is unsuccessful with some probability when $t>\floor{\frac{d-1}{2}}$. The focus of this work is to bound this success probability, assuming an uniform distribution of errors of given weight $t$. The techniques we use are based on analyzing~\cref{eq:KeyEquationLSE} and overbounding the number of cases where $\rank(\S(t))<t$ when $t$ errors occurs. To bound the success probability of~\cref{algo:SyndromeDecoder} based on this analysis, we first show that $\rank(\S(t))<t$ is a \emph{necessary and sufficient} condition for \cref{algo:SyndromeDecoder} to be unsuccessful. In other words, as $\rank(\S(t))\leq t$ by design, the decoder succeeds exactly when $\S(t)$ is of full rank $t$. %
 The arguments are an extension of those in the proof of \cite[Lemma~2]{schmidt2009collaborative}.
 \begin{lemma}[Condition for unsuccessful decoding]\label{lem:FailRankCond}
   Let $\cI\code^{(\ell)}$ be an $\ell$-interleaved alternant code with $\code \in \mathbb{A}^d_{\va}$. For a received word $\R=\C+\widetilde{\E}$, where $\C\in\cI\code^{(\ell)}$ and the error $\widetilde{\E}$ has $t>0$ nonzero columns, \cref{algo:SyndromeDecoder} is not \emph{successful}, i.e., returns $\hat \C \ne \C$ or a {\normalfont \texttt{decoding failure}}, if and only if $\rank(\S(t))<t$.
\end{lemma}
\begin{proof}
  Denote by $\Lambda(x)$ the \emph{true} error locator polynomial corresponding to the $t$ error positions (indices of non-zero columns) in $\widetilde{\E}$. Then $\Lambda(x)$ has $t$ distinct roots in $\Fqm$ and $\vLambda$ is a solution of ${\frak S}(t)$ as in~\cref{eq:KeyEquationLSE}.

\emph{Necessary condition:} We show that unsuccessful decoding implies $\rank(\S(t))<t$.

The algorithm can fail only on \cref{step:uniqueSolution} or \ref{step:distinctRoots}. \cref{step:solveLSE} determines the \emph{minimal} $t^{\star}$ such that ${\frak S}(t^{\star})$ has at least one solution $\vLambda^{\star}$, hence $t^{\star} \leq t$. Note that $\vLambda^{\star}$ is also a solution to ${\frak S}(t)$ since $t\ge t^{\star}$ (cf.~\cite[Lemma~2]{schmidt2009collaborative}).
 If the algorithm fails on \cref{step:uniqueSolution}, the system ${\frak S}(t^\star)$ has many solutions, hence ${\frak S}(t)$ also has many solutions and $\rank(\S(t))<t$.
A failure on \cref{step:distinctRoots} occurs if $\Lambda^{\star}(x)$ does not have $t^{\star}$ different roots which implies $\Lambda^{\star}(x)\ne \Lambda(x)$. Again, the system ${\frak S}(t)$ has at least two solutions $\vLambda$ and $\vLambda^{\star}$ and $\rank(\S(t))<t$.

Only \cref{step:zeroSyndromes} and \ref{step:codewordEstimate} can result in a miscorrected codeword. If the decoder outputs $\hat \C$ on \cref{step:zeroSyndromes}, we have $\hat \C \ne \C$ as $t>0$. Further, in this case $\S(t) = \0$, so $\rank(\S(t))= 0 <t$. %
If the algorithm outputs a miscorrected codeword $\hat \C \ne \C$ on \cref{step:codewordEstimate}, the error positions in $\R - \hat \C$ correspond to a $\Lambda^{\star}(x)$ whose coefficients $\vLambda^{\star}$ are a solution to ${\frak S}(t^\star)$ and hence also to ${\frak S}(t)$. Thus ${\frak S}(t)$ has two different solutions $\vLambda^{\star}$ and $\vLambda$, which are different since $\hat \C \ne \C$, and it follows that $\rank(\S(t))<t$.

\emph{Sufficient condition:} We show that unsuccessful decoding follows from $\rank(\S(t))<t$.

Only \cref{step:zeroSyndromes} or \ref{step:codewordEstimate} can result in the output of a valid codeword. Let us assume that $\rank(\S(t))<t$ but the decoding was successful, i.e., $\hat \C = \C$. If $\C$ was found in \cref{step:zeroSyndromes} then $\R=\C$ and the number of errors is $t=0$, which contradicts the assumption $t>0$.
If the correct $\C=\hat \C$ was the result of \cref{step:codewordEstimate}, then the \emph{minimal} $t^\star$ is equal to the \emph{actual} number of errors $t$ and $\vLambda^\star=\vLambda$;
otherwise it is not possible for the polynomial $\Lambda^\star(x)$, which is of degree $t^{\star}$, to have $t$ distinct roots.
Since, by assumption, the algorithm did not fail, it follows from \cref{step:uniqueSolution} that in this case ${\frak S}(t)$ has a unique solution which contradicts our assumption that $\rank(\S(t))<t$.

\end{proof}

\begin{remark}[Application of \cref{lem:FailRankCond} to interleaved RS codes]\label{rem:applicationToRS}
In \cite[Lemma~2]{schmidt2009collaborative} it is proved that \cref{algo:SyndromeDecoder} returning a {\normalfont \texttt{decoding failure}} is a sufficient condition for the matrix $\S(t)$ to be rank deficient. Therefore, an upper bound on the probability of $\rank(\S(t))<t$ provides an upper bound on the probability of a decoding failure. In \cref{lem:FailRankCond} we extend this argument by showing that the decoder does not succeed \emph{if and only if} $\rank(\S(t))<t$. This implies that any bound on the probability of $\rank(\S(t))<t$ is not only a bound on the probability of a decoding failure, but an upper bound on sum of the probability of a decoding failure and the probability that the decoder returns a miscorrection. As this is a property of the decoder and therefore not specific to interleaved alternant codes, it follows that the upper bound on the probability of a decoding failure for interleaved RS codes of \cite[Theorem~7]{schmidt2009collaborative} is in fact an upper bound on the probability of the decoder being unsuccessful, i.e., a bound on $1-\Psuc$.
\end{remark}

With the help of~\cref{lem:FailRankCond}, we now present the crux in bounding the success probability, which is the basis of the bounds presented in \cref{sec:upperBounds}.

\begin{lemma}\label{lem:FailureCrux}
  Let $\mathcal{IC}^{(\ell)}$ be an $\ell$-interleaved alternant code with $\code \in \mathbb{A}^d_{\va}$ and $\cE=\{j_1,j_2,\dots,j_t\}\subset [n]$ be a set of $|\cE|=t$ error positions, where $n\coloneqq |\va|$. For a codeword $\C\in \mathcal{IC}^{(\ell)}$, an error matrix $\widetilde{\E} \in \Fq^{\ell \times n}$ with $\supp(\widetilde{\E}) \coloneqq \cE$ and $\E \coloneqq \widetilde{\E}|_{\cE} \in \EB{q}{\ell}{t}$, and a received word $\R \coloneqq \C + \widetilde{\E}$, \cref{algo:SyndromeDecoder} \emph{succeeds}, i.e., returns $\hat \C = \C$, if and only if
  \begin{equation}
    \label{eqn:failure_cond}
    \nexists \v \in \Fqm^t \setminus \{\0\} \text{ such that } \H \cdot \diag(\v) \cdot \E^\top = \0 \ , %
  \end{equation}
  where
  $\H\in \F_{q^m}^{d-t-1 \times t}$ is a parity-check matrix of the code $\GRS_{\va|_\cE,\1}^{d-t}$.
\end{lemma}
\begin{proof}
  We extend and adapt the proof for interleaved RS codes from~\cite{schmidt2009collaborative}.

  According to~\cref{lem:FailRankCond}, \cref{algo:SyndromeDecoder} may only yield a \texttt{decoding failure} or a miscorrection $\hat \C \neq \C$ if $\rank(\S(t))<t$, with $\S(t)$ as in \cref{eq:KeyEquationLSE}. In other words, the decoding may only be unsuccessful, if there exists a non-zero vector $\u\in \Fqm^t$ such that $\S(t) \cdot \u=\0$, i.e.,
  \begin{align}\label{eq:ustatement}
    \exists \u\in \Fqm^t\setminus\{\0\} \text{ such that }\supbrac{\S}{i}(t)\cdot \u=\0\ ,\ \forall i\in[\ell]\ .
  \end{align}
  It is known (cf.~\cite[Theorem 9.9]{PetersonWeldon1972}\cite{schmidt2009collaborative}) that a syndrome matrix $\supbrac{\S}{i}(t)$ can be decomposed into
  \begin{align*}
    \supbrac{\S}{i}(t)=\H \cdot \supbrac{\Fbold}{i}\cdot \D\cdot \V\ ,
  \end{align*}
  where the matrix $\H$ is defined as in the statement (see also \cref{def:GRScodes}),
  \begin{align*}
    &\V=
    \begin{pmatrix}
        1 & 1 & \dots & 1\\
      \alpha_{j_1} & \alpha_{j_2} & \dots & \alpha_{j_t}\\
      \alpha_{j_1}^2 & \alpha_{j_2}^2 & \dots & \alpha_{j_t}^2\\
      \vdots & \vdots &  & \vdots\\
      \alpha_{j_1}^{t-1} & \alpha_{j_2}^{t-1} & \dots & \alpha_{j_t}^{t-1}
    \end{pmatrix}^\top\in\Fqm^{t\times t} \ , \\
    &\supbrac{\Fbold}{i} =\diag(\E_{i,:}) \in \Fq^{t\times t}\ ,\\
    &\D = \diag(\v'|_{\cE})\in \Fqm^{t \times t}\ ,
  \end{align*}
 and $\v'$ is the column multiplier of the GRS code corresponding to the alternant code $\code$, i.e., $\GRS^d_{\va,\v'}\cap\Fq=\code$.

  We observe that the matrices $\D$ and $\V$ are both square and of full rank. Therefore, the product $\v=\D\cdot\V\cdot\u$ defines a one-to-one mapping $\u\to\v$, such that $\0\to\0$. Consequently, the statement~\cref{eq:ustatement} is equivalent to the statements
  \begin{align*}
    \exists \v \in \Fqm^t \setminus \{\0\} \text{ such that } \H&\cdot \diag(\E_{i,:})\cdot \v =\0\ ,\ \forall i\in[\ell]\\
    &\Updownarrow\\
    \exists \v \in \Fqm^t \setminus \{\0\} \text{ such that } \H&\cdot \diag(\v)\cdot \E_{i,:} =\0\ ,\ \forall i\in[\ell]\ ,
  \end{align*}
  and the lemma statement follows.
\end{proof}

Above we extended and adapted the first part of the proof of the upper bound on the failure probability for interleaved RS codes in~\cite{schmidt2009collaborative}, where the error matrix $\widetilde{\E}$ is assumed to be over $\Fqm$ (the field of RS codes). Simulation results indicate that this bound is quite tight. However, for interleaved alternant codes, $\widetilde{\E}$ is over $\Fq$ (the \emph{subfield} of RS codes) and the bound from~\cite{schmidt2009collaborative} is not valid in this case.

\cref{lem:FailureCrux} gives a necessary and sufficient condition for~\cref{algo:SyndromeDecoder} to succeed for an error $\widetilde{\E}$ with fixed $\cE=\supp(\widetilde{\E})$ and $\widetilde{\E}|_{\cE}\in \EB{q}{\ell}{t}$.
In~\cref{sec:upperBounds} and~\cref{sec:lowerBound} we bound the probability of successful decoding of~\cref{algo:SyndromeDecoder} for a random error matrix $\widetilde{\E}$ where $\widetilde{\E}|_{\cE}\sim \EB{q}{\ell}{t}$. %

\section{Technical Preliminary Results}\label{sec:TechnicalPreliminaries}

Before deriving the bounds on the success probability of decoding interleaved alternant codes in~\cref{thm:failureProbNewBoundGeneral}, we establish some technical preliminary results which are needed in the proof of the bounds.

\subsection{Maximization of Integer Distributions}

To begin, we derive a simple upper bound on the maximization of a sum of integer powers, under a restriction on the base of the power.

\begin{definition}[Majorization Relation]\label{def:majorization}
Let $\cM=\{\{M_1, M_2,\dots,M_c\}\}$ and $\cK=\{\{K_1, K_2,\dots,K_c\}\}$ be two (finite) multi-sets of real numbers with the same cardinality. We say that the set $\cM$ \emph{majorizes} the set $\cK$ and write
\begin{equation*}
\cM\succ\cK\quad\mathrm{or}\quad \cK\prec\cM
\end{equation*}
if, after a possible renumeration, $\cM$ and $\cK$
satisfy the following conditions:
\begin{enumerate}
    \item[(1)] $M_1\geq M_2 \geq \cdots \geq M_c$ and $K_1\geq K_2 \geq \cdots \geq K_c$;
    \item[(2)] $\sum_{i=1}^jM_i \geq \sum_{i=1}^j K_i, \ \forall \, 1\leq j \leq c$;
\end{enumerate}
\end{definition}

\begin{lemma}[Karamata's inequality~{\cite[Theorem~1]{kadelburg2005inequalities}}] \label{lem:karamata}
Let $\cM=\{\{M_1, M_2,\dots,M_c\}\}$ and $\cK=\{\{K_1, K_2,\dots,K_c\}\}$ be two multi-sets of real numbers from an interval $[a,b]$. If the set $\cM\succ\cK$, and if $f:[a,b]\to \mathbb{R}$ is a convex and non-decreasing function, then it holds that
\begin{align}
    \sum\limits_{i=1}^c f(M_i)\geq \sum\limits_{i=1}^c f(K_i) \ .
\end{align}
\end{lemma}

For convenience of notation, we define a fixed notation for the set over which we will maximize in the following.
\begin{definition} \label{def:multiset}
Denote by $\mathbb{M}_{c,B}^{(a,b)} = \{\cM,\ldots\}$ the set of all multi-sets $\cM = \{\{M_1,\ldots, M_c\}\}$ of cardinality $c$ with $b\geq M_1 \geq \ldots \geq M_c\geq a $ and $\sum_{M\in \cM} M = B$.
\end{definition}

With these definitions established, we are now ready to give an upper bound on the sum over the results of a convex non-decreasing function evaluated on the elements of any multi-set in $\mathbb{M}_{c,B}^{(a,b)}$.

\begin{lemma}\label{lem:maximization}
  Let $a,c\geq 1$, $b\geq a$, $ca\leq B \leq c b$, and $\mathbb{M}_{c,B}^{(a,b)}$ be as in \cref{def:multiset}. Then, for any function $f(x)$ that is convex and non-decreasing in the interval $a\leq x \leq b$, it holds that
  \begin{align*}
    \max\limits_{\cM\in\mathbb{M}_{c,B}^{(a,b)}} \sum_{M\in \cM} f(M) \!\leq\! \Big(\frac{B-ca}{b-a}+1\Big)( f(b) - f(a)) + cf(a)
  \end{align*}
\end{lemma}
\begin{proof}
  By definition
  \begin{align*}
\sum\limits_{M\in\cM}M=\sum\limits_{x \in \cM}\delta^x_{\cM}x=B, \ \forall \ \cM\in\mathbb{M}_{c,B}^{(a,b)}
  \end{align*}
  and it follows that for all $\cM\in\mathbb{M}_{c,B}^{(a,b)}$ we have
  \begin{align*}
    \delta^b_{\cM}  &= \frac{1}{b}\bigg({B - \sum\limits_{x \in \cM\setminus\{b\}}\delta^x_{\cM}x}\bigg) \leq \frac{B-(c-\delta^b_{\cM})a }{b}, \\
    \delta^b_{\cM} &\leq \frac{B-ca}{b-a}.
  \end{align*}
  Let $\cM_{\max} = \{b,\dots,b,a,\dots,a\}$ be a multiset with $\delta^b_{\cM_{\max}}=\ceil{\frac{B-ca}{b-a}}$ and $\delta^a_{\cM_{\max}} = c-\delta^b_{\cM_{\max}}$. It can readily be seen that $\cM_{\max}\succ\cM \ \forall \ \cM\in\mathbb{M}_{c,B}^{(a,b)}$ (note that $\cM_{\max}\in \mathbb{M}_{c,B}^{(a,b)}$ if $(b-a)|(B-ca)$).

Since $f(x)$ is a convex non-decreasing function for $a\leq x \leq b$, it follows from \cref{lem:karamata} that
\begin{equation}\label{eq:sum_inequality_karamata}
\sum\limits_{M\in \cM_{\max}} f(M) \geq \sum\limits_{M\in \cM} f(M) \ , \ \forall\ \cM \in \mathbb{M}_{c,B}^{(a,b)}\ .
\end{equation}

For $\cM_{\max}$ we have
\begin{align*}
  \max_{\cM \in \mathbb{M}_{c,B}^{(a,b)}} \sum_{M\in \cM} f(M) &\leq \sum\limits_{M\in \cM_{\max}} f(M) \\
  &= \delta^b_{\cM_{\max}} f(b) + (c-\delta^b_{\cM_{\max}})f(a) \\
  &=  \ceil{\frac{B-ca}{b-a}}( f(b) - f(a)) + cf(a)
\end{align*}
and the lemma statement follows.
\end{proof}

\subsection{Sum over the Cardinalities of Alternant Codes}

Specific subclasses of alternant codes, such as some BCH and Goppa codes, are known to have larger dimension \cite{macwilliams1977theory} than the lower bound given in \cref{lem:AlternantDimBounds}. However, in general it is a difficult and open problem to predict the dimension of an alternant code for given column multipliers $\v$. On the other hand, the sum over the cardinality of subfield subcodes for all combinations of non-zero column multipliers is easily determined, not only for alternant codes, but for any linear code with a known weight distribution.

For a linear $[n,k,d]_{q^m}$ code $\code$, define
\begin{align*}
  B_{n,d,w}(\code) &\coloneqq \sum_{\v \in (\Fqm^{\star})^n} \Big\lvert \{ \c \cdot \diag(\v) \ | \ \c \in \code, \wt(\c) = w\} \cap \Fq^n \Big\rvert \ .
\end{align*}
Since every linear code contains the all-zero codeword and no other codeword of weight $<d$, the sum over the cardinality of the subcodes for all combinations of non-zero column multipliers is given by
\begin{align*}
  B_{n,d}(\code) &\coloneqq  (q^m-1)^n + \sum_{w=d}^{n} B_{n,d,w}(\code) \\
  &= \sum_{\v \in (\Fqm^{\star})^n} \Big\lvert \{ \c \cdot \diag(\v) \ | \ \c \in \code\} \cap \Fq^n\Big\rvert \ .
\end{align*}
Observe that if $\code$ is a $\GRS_{\va,\v'}^d$ code for some $\v' \in (\Fqm^{\star})^n$, then $B_{n,d,w}$ is the sum over the number of codewords of weight $w$ in all alternant codes $\ALTallp$ and $B_{n,d}(\code)$ is the sum over their cardinalities. Interestingly, while the weight enumerators and cardinality of a specific subfield subcode depend on $\v$, the sum of these values over all $\v$ only depends on the weight enumerators of $\code$.

\begin{lemma}\label{lem:sumCardinalitiesAlternant}
Let $\code$ be an $[n,k,d]_{q^m}$ code and denote by $A_w^{\code}$ the $w$-th weight enumerator of $\code$. Then,
\begin{align*}
  B_{n,d,w}(\code) = A^{\code}_{w}\cdot (q^m-1)^{n-w}(q-1)^w \ .
\end{align*}
\end{lemma}
\begin{proof}
Let $\c$ be a codeword of $\code$. We have $\c \cdot \diag(\v) \in \Fq^n$ if and only if $c_iv_i \in \Fq$ for all $i \in [n]$. If $i \in \supp(\c)$, then there are exactly $q-1$ choices of $v_i$ for which $c_iv_i \in \Fq$. Else, any of the $q^m-1$ possible values of $v_i$ give $c_iv_i=0 \in \Fq$. Hence, we have
\begin{align*}
B_{n,d,w}(\code)\! &= \!\! \sum_{\v \in (\Fqm^{\star})^n}\! \Big\lvert \{ \c \cdot \diag(\v) \ | \ \c \in \code, \wt(\c) = w\} \cap \Fq^n \Big\rvert \\
 &= \sum_{\substack{\c \in \code \\ \wt(\c) = w}} \Big\lvert \{ \v \in (\Fqm^{\star})^n \ | \ c_iv_i \in \Fq \ \forall i \in [n] \}  \Big\rvert \\
  &= A^{\code}_{w}\cdot (q^m-1)^{n-w}(q-1)^w \ .
\end{align*}
\end{proof}

The weight distribution of an MDS code only depends on its parameters, not the code itself (see~\cref{thm:MDSWeightEnumerator}). Hence, for an MDS code $\code$ we can omit the dependence on $\code$ and write
\begin{equation}\label{eq:BforMDS}
  B^{\mathsf{MDS}}_{n,d,w} \coloneqq B_{n,d,w}(\code)  \ \text{and} \ B^{\mathsf{MDS}}_{n,d} \coloneqq B_{n,d}(\code)\ .
\end{equation}

\subsection{Probability of a Code Containing a Random Matrix}

We begin by proving a technical lemma that bounds the probability that all rows of a randomly chosen matrix with no all-zero columns are in a code of a certain dimension. This is a refined version of \cite[Lemma~3]{schmidt2009collaborative}. %
\begin{lemma}\label{lem:Pwk}
For some integers $\ell>0,n\geq k\geq 0$, let $\cA$ be an $[n,k]_q$ code and denote by $A_w^{\cA}$ its $w$-th weight enumerator.
Then, for $\E \sim \EB{q}{\ell}{n}$ we have
\begin{align*}
  \Pr_{\E}\{\E_{i,:}& \in\cA \ \forall i \in [\ell]\} \\
  &\leq \frac{q^{k\ell}(q-1)- (q^\ell-1)(q^k-1-A_n^{\cA}) -(q-1) }{(q-1)(q^\ell-1)^n} \ .
\end{align*}
\end{lemma}
\begin{proof}
  Let $\cL \subset \F_q^{\ell\times n}$ the set of matrices whose rows are codewords of $\cA$ and by $\cL_0\subset \cL$ the subset of all matrices in $\cL$ \emph{with} at least one all-zero column. Observe that
\begin{equation*}
  \left\{\E \ \left| \
  \begin{array}{l}
    \E_{1,:},\ldots,\E_{\ell,:} \ \text{are $\Fq$-scalar multiples of} \ \e ,\\
    \e\in\bar{\cA}\cup\{0\},\wt(\e)<n
  \end{array} \right.
\right\} \subseteq \cL_{0}\ ,
\end{equation*}
where $\bar{\cA}$ is a set of representatives\footnote{A common choice is the set of all non-zero codewords of $\cA$ whose first non-zero entry is $1$.} of $(\cA\setminus\{0\})/\Fq^{\star}$, which has cardinality $|\bar{\cA}|=\frac{q^k-1}{q-1}$.
If $\mathbf{e}=\mathbf{0}$ there is only one matrix, i.e., the all-zero matrix. For all other $\e$ with $\wt(\e)<n$ each row can be an $\Fq$-multiple of $\e$ and all these matrices are unique, if at least one row is not $\0$. The number of such choices is $q^\ell-1$, so
\begin{align*}
    |\cL_0|&\geq (q^\ell-1) (|\bar{\cA}|-\underbrace{|\{\c\in\bar{\cA} \ | \ \wt(\c)=n\}|}_{\eqqcolon\frac{A_n^{\cA}}{(q-1)}}) + 1 \\
    &= \frac{(q^\ell-1)}{(q-1)}(q^k-1-A_n^{\cA}) +1 \ .
\end{align*}
Recall that $\EB{q}{\ell}{n}$ does not contain any matrices with all-zero columns by definition, so $\cL_0\cap \EB{q}{\ell}{n}= \emptyset$. As $\cL_0 \subset \cL$, it follows that
\begin{align*}
  \Pr_{\E}\{\E_{i,:} \in\cA \ \forall i=[\ell]\}&=\frac{|\cL\cap\EB{q}{\ell}{n}|}{|\EB{q}{\ell}{n}|} \\
  &= \frac{|\cL\setminus\cL_{0}|}{|\EB{q}{\ell}{n}|} = \frac{|\cL|-|\cL_0|}{|\EB{q}{\ell}{n}|} \ .
\end{align*}
The lemma statement follows from the observation that $|\cL|=|\cA|^\ell=q^{k\ell}$ and $|\EB{q}{\ell}{n}|=(q^\ell-1)^n$.

\end{proof}

If $|\cL_0|$ is large, it is worthwhile to deduct it from $|\cL|$ as in \cref{lem:Pwk}. However, for other parameters, (our best lower bound on) $|\cL_0|$ becomes negligible compared to $|\cL|$. Therefore, we also define a simplified version of this upper bound, where we only exclude the zero matrix from $\cL$.

\begin{corollary}\label{cor:PwkL01}
For some integers $\ell>0,n\geq k\geq 0$, let $\cA$ be an $[n,k]_q$ code.
Then, for $\E \sim \EB{q}{\ell}{n}$ we have
\begin{equation*}
    \Pr_{\E}\{\E_{i,:} \in\cA \ \forall i\in[\ell]\} \leq  \frac{|\cL\setminus\{\0_{\ell\times n}\}|}{|\EB{q}{\ell}{n}|} = \frac{q^{k\ell} -1 }{(q^\ell-1)^n} \ .
\end{equation*}
\end{corollary}

\section{The Success Probability of Decoding Interleaved Alternant Codes}\label{sec:upperBounds}

We now turn to the main topic of this work, namely providing bounds on the performance of the decoder of \cite{FengTzeng1991ColaDec,schmidt2009collaborative} (see~\cref{sec:decoding_details}) when applied to interleaved alternant codes. Recall that the success probability is given by
\begin{align*}
  \Psuc = 1- \Pfail - \Pmisc \ ,
\end{align*}
where $\Pfail$ and $\Pmisc$ are the probability of a decoding failure and a miscorrection, respectively. %

We begin by applying the technical results of \cref{sec:TechnicalPreliminaries} to obtain a lower bound on the success probability of decoding interleaved alternant codes that is valid for any interleaving order $\ell$. The applied principle is a generalization of the approach in \cite{schmidt2009collaborative}.

\subsection{A Lower Bound on the Success Probability for any Interleaving Order $\ell$}

To begin, we relate the problem of bounding the probability successful decoding to properties of the multisets $\ALTallp$ of alternant codes for different parameters.

\begin{theorem}\label{thm:failureProbNewBoundGeneral}
  Let $\mathcal{IC}^{(\ell)}$ be an $\ell$-interleaved alternant code with $\code \in \mathbb{A}^d_{\va}$ and $\cE=\{j_1,j_2,\dots,j_t\}\subset [n]$ be a set of $|\cE|=t$ error positions, where $n\coloneqq |\va|$. For a codeword $\C\in \mathcal{IC}^{(\ell)}$, an error matrix $\widetilde{\E} \in \Fq^{\ell \times n}$ with $\supp(\widetilde{\E}) \coloneqq \cE$ and $\E \coloneqq \widetilde{\E}|_{\cE} \sim \EB{q}{\ell}{t}$, and a received word $\R \coloneqq \C + \widetilde{\E}$, \cref{algo:SyndromeDecoder} \emph{succeeds}, i.e., returns $\hat \C = \C$, with probability
\begin{align*}
  &\Psuc(\mathcal{IC}^{(\ell)},\cE) \geq 1- \\
  &\ \ \sum_{w=d-t}^t \sum_{\substack{\cV \subseteq [\cE]\\ |\cV|=w}} \sum_{\cA \in \mathbb{A}_{\va|_{\cV}}^{d-t}} \!\! \Big(\delta_{\mathbb{A}_{\va|_{\cV}}^{d-t}}^{\cA}\Big)^{-1} \underset{\E}{\Pr}\{ (\E|_{\cV})_{i,:} \in \cA \ \forall \ i\in [\ell]  \}\ ,
\end{align*}
where $\delta_{\mathbb{A}_{\va|_{\cV}}^{d-t}}^{\cA}$ is the multiplicity of $\cA$ in $\mathbb{A}_{\va|_{\cV}}^{d-t}$.
\end{theorem}
\begin{proof}
  Denote $\E \coloneqq \widetilde{\E}|_{\cE}$. By \cref{lem:FailureCrux} the decoding of $\widetilde{\E}$ succeeds if and only if%
  \begin{equation*}
    \nexists \v \in \Fqm^t \setminus \{\0\} \text{ such that } \H \cdot \diag(\v) \cdot \E^\top = \0 \ ,
  \end{equation*}
where $\H\in\Fqm^{(d-t-1)\times t}$ denotes the parity-check matrix of the code $\GRS_{\va|_\cE, \boldsymbol{1}}^{d-t}$, i.e., the RS codes of distance $d-t$ with locators corresponding to the error positions.

Therefore, the probability of unsuccessful decoding is upper bounded by
\begin{align}
  1&-\Psuc(\mathcal{IC}^{(\ell)},\mathcal{E}) \\
  &\leq \underset{\E}{\Pr}\{\exists \ \v \in \Fqm^{t} \setminus \{\0\} : \H \cdot \diag(\v) \cdot \E^\top = \0 \}\nonumber \\
                     &\leq \sum_{w=1}^{t} \underset{\E}{\Pr}\{\exists \ \v \in \Fqm^{t}, \wt(\v) = w :\H \cdot \diag(\v) \cdot \E^\top = \0 \}\label{eq:ieq1}\\
                     &\stackrel{\mathsf{(a)}}{=} \! \sum_{w=d-t}^t \! \underset{\E}{\Pr}\{ \exists \v \in \Fqm^{t}, \wt(\v) \!=\! w : \H \cdot \diag(\v) \cdot \E^\top \! = \0 \}\nonumber\\
                     &=  \sum_{w=d-t}^t \sum_{\substack{\cV \subseteq [\cE]\\ |\cV|=w}} \underset{\E}{\Pr}\{ \exists \ \cA \in \mathbb{A}_{\va|_\cV}^{d-t} : (\E|_{\cV})_{i,:} \in \cA \ \forall \ i\in [\ell]  \}\nonumber\\
                     &\leq  \sum_{w=d-t}^t \sum_{\substack{\cV \subseteq [\cE]\\ |\cV|=w}} \sum_{\cA \in \mathbb{A}_{\va|_{\cV}}^{d-t}} \! \! \Big(\delta_{\mathbb{A}_{\va|_{\cV}}^{d-t}}^{\cA}\Big)^{-1} \underset{\E}{\Pr}\{ (\E|_{\cV})_{i,:} \! \in\! \cA  \forall  i\in [\ell]  \} \ , \nonumber
\end{align}
where $\mathsf{(a)}$ holds because any $d-t-1$ columns of $\H$ are linearly independent.
\end{proof}

With this connection between the multisets $\mathbb{A}_{\va|_{\cV}}^{d-t}$ and the probability of successful decoding $\Psuc(\mathcal{IC}^{(\ell)},\cE)$ established, we now apply the technical results of \cref{sec:TechnicalPreliminaries} to obtain a lower bound.

\begin{theorem}\label{thm:failureProbNewBound}
The probability of successful decoding $\Psuc(\cI\code^{(\ell)},\cE)$ as in \cref{thm:failureProbNewBoundGeneral} is lower bounded by
\begin{align*}
  \Psuc(&\mathcal{IC}^{(\ell)},\cE) \geq 1-\sum_{w=d-t}^t \frac{\binom{t}{w}}{(q^m-1)(q^\ell-1)^w}\\
  & \quad \cdot \bigg( \frac{(q^\ell-1)}{(q-1)}\Big(c_w+B_{w,d-t,w}^{\mathsf{MDS}}-B_{w,d-t}^{\mathsf{MDS}}\Big) -c_w \\
  & \quad \quad \ \ + \Big(\frac{B_{w,d-t}^{\mathsf{MDS}}-c_w a_w}{b_w-a_w}+1\Big)( b_w^\ell - a_w^\ell) + c_w a_w^\ell \bigg) \ ,
\end{align*}
with
\begin{align*}
    a_w &= \max\{1,q^{w-(d-t-1)m}\}, \\
    b_w &= q^{k_q^{\mathsf{opt.}}(w,d-t)}, \\
    c_w &= (q^m-1)^{w},
\end{align*}
where $B_{w,d-t}^{\mathsf{MDS}}$ and $B_{w,d-t,w}^{\mathsf{MDS}}$ are given in \cref{eq:BforMDS}, $k_q^{\mathsf{opt.}}(w,d-t)$ is an upper bound on the dimension of a $q$-ary code of length $w$ and minimum distance $d-t$.%
\end{theorem}
\begin{proof}
For a $q$-ary code $\cA$ denote $ k_{\cA} \coloneqq \dim_q(\cA)$.
\newcounter{storeeqcounter}
\newcounter{tempeqcounter}
Starting from \cref{thm:failureProbNewBoundGeneral}, we obtain the derivation given on the top of %
Page~\pageref{flteq:proofOfNewBound},
\begin{figure*}[!t]
\normalsize
  \begin{align*}
  1-\Psuc(\mathcal{IC}^{(\ell)}, \mathcal{E}) &\leq  \sum_{w=d-t}^t \sum_{\substack{\cV \subseteq [\cE]\\ |\cV|=w}} \sum_{\cA \in \mathbb{A}_{\va|_{\cV}}^{d-t}} (\delta_{\mathbb{A}_{\va|_{\cV}}^{d-t}}^{\cA})^{-1} \ \underset{\E}{\Pr}\{ (\E|_{\cV})_{i,:} \in \cA \ \forall \ i\in [\ell]  \}  \\
                     &\stackrel{\mathsf{(a)}}{\leq}  \sum_{w=d-t}^t \sum_{\substack{\cV \subseteq [t]\\ |\cV|=w}} \sum_{\cA \in \mathbb{A}_{\va|_{\cV}}^{d-t}} (q^m-1)^{-1} \frac{(q-1)q^{k_\cA \ell}-(q^\ell-1)(q^{k_\cA}-1-A^\cA_w)-(q-1) }{(q-1)(q^\ell-1)^w}\nonumber\\
                     &\stackrel{\mathsf{(b)}}{=} \sum_{w=d-t}^t \sum_{\substack{\cV \subseteq [t]\\ |\cV|=w}}\frac{1}{(q^m-1)(q^\ell-1)^w} \left(\frac{(q^\ell-1)}{(q-1)}(c_w + B_{w,d-t,w}^{\mathsf{MDS}})-c_w + \left( \sum_{\cA \in \mathbb{A}_{\va|_{\cV}}^{d-t}} q^{k_\cA\ell}-\frac{(q^\ell-1)}{(q-1)}q^{k_\cA} \right)\right)\nonumber\\
                     &\stackrel{\mathsf{(c)}}{\leq} \sum_{w=d-t}^t \frac{\binom{t}{w}}{(q^m-1)(q^\ell-1)^w} \left(\frac{(q^\ell-1)}{(q-1)}(c_w + B_{w,d-t,w}^{\mathsf{MDS}})-c_w + \!\! \max_{ \substack{\cM \in \mathbb{M}_{c_w}^{[a_w,b_w]} \\ \sum_{M \in \cM} M = B^{\mathsf{MDS}}_{w,d-t}}}\! \sum_{M \in \cM} M^\ell-\frac{(q^\ell-1)}{(q-1)}M\right)\nonumber\\
                     &= \sum_{w=d-t}^t \frac{\binom{t}{w}}{(q^m-1)(q^\ell-1)^w} \left(\frac{(q^\ell-1)}{(q-1)}(c_w + B_{w,d-t,w}^{\mathsf{MDS}}- B^{\mathsf{MDS}}_{w,d-t})-c_w +\max_{ \mathcal{M} \in \mathbb{M}_{c_w,B_{w,d-t}}^{[a_w,b_w]}} \sum_{M \in \cM} M^\ell\right) \nonumber
\end{align*}
\label{flteq:proofOfNewBound}
\hrulefill
\vspace*{4pt}
\end{figure*}
where $\mathsf{(a)}$ holds by \cref{eq:multiplicityAlternant} and \cref{lem:Pwk}, $(\mathsf{b})$ holds as $\sum_{\cA \in \mathbb{A}_{\va|_{\cV}}^{d-t}}A^\cA_w=B_{w,d-t,w}^{\mathsf{MDS}}$ (see~\cref{eq:BforMDS}) and $|\mathbb{A}_{\va|_{\cV}}^{d-t}|=c_w$ (see~\cref{eq:cardinalityAall}), and $\mathsf{(c)}$ holds as $a_w$ and $b_w$ are lower and upper bounds on the cardinality of all codes $\cA \in \mathbb{A}_{\va|_{\cV}}^{d-t}$ (see~\cref{lem:AlternantDimBounds}) and because $\sum_{\cA \in \mathbb{A}_{\va|_{\cV}}^{d-t}} q^{k_\cA} = B_{w,d-t}^{\mathsf{MDS}}$ by \cref{lem:sumCardinalitiesAlternant}. The theorem statement follows by \cref{lem:maximization}.
\end{proof}

By the use of \cref{cor:PwkL01} instead of \cref{lem:Pwk} in step $\mathsf{(a)}$ we get a slightly simplified (though worse) lower bound.
\begin{corollary}\label{cor:failureProbNewBoundL0}
The probability of successful decoding $\Psuc(\code,\cE)$ as in \cref{thm:failureProbNewBoundGeneral} is lower bounded by
\begin{align*}
  \Psuc(&\mathcal{IC}^{(\ell)},\cE) \geq 1-\sum_{w=d-t}^t \frac{\binom{t}{w}}{(q^m-1)(q^\ell-1)^w} \\
  &\ \ \cdot \Big(  -c_w+\Big(\frac{B_{w,d-t}^{\mathsf{MDS}}-c_w a_w}{b_w-a_w}+1\Big)( b_w^\ell - a_w^\ell) + c_w a_w^\ell \Big)
\end{align*}
with
\begin{align*}
    a_w &= \max\{1,q^{w-(r-t)m}\} \\
    b_w &= q^{k_q^{\mathsf{opt.}}(w,d-t)} \\
    c_w &= (q^m-1)^{w},
\end{align*}
where $B_{w,d-t}^{\mathsf{MDS}}$ is given in \cref{eq:BforMDS}, $\kopt(w,d-t)$ is an upper bound on the dimension of a $q$-ary code of length $w$ and minimum distance $d-t$.%

\end{corollary}

\subsection{A Lower Bound on the Success Probability for any Interleaving Order  $\ell\geq t$} \label{sec:largeEll}

For large interleaving order $\ell \geq t$, the Metzner-Kapturowski generic decoder~\cite{metzner1990general} guarantees to decode any $1\leq t\leq d-2$ errors if $\rank(\E)=t$ in an $\ell$-interleaved code with \emph{any} $[n,k,d]_q$ constituent code. The decoder has been generalized in~\cite{haslach2000efficient} for the case of rank deficiency when $2t-d+2 \leq \rank(\E)< t$. However, if the structure of the constituent code is unknown, determining the error positions in a rank-deficient error matrix $\E$ where $\rank(\E)=\mu<t$ is equivalent to finding a subset $\cU$ of columns of a parity-check matrix $\H\in\Fqm^{(d-1-\mu)\times n}$ with $\rank(\H|_{\cU})=t-\mu$. This is known to be a hard problem and no polynomial-time algorithm is known if the rank deficiency $t-\mu$ becomes large~\cite{roth2014coding}. If the code structure is given, efficient syndrome-based algorithms are proposed in~\cite{roth2014coding} and~\cite{YuLoeliger16} to correct linearly dependent error patterns with $\rank(\E)\geq 2t-d+2$ by interleaved RS codes over $\Fqm$.
These decoders also apply to the class of alternant codes over $\Fq$.
Consider an $\ell$-interleaved alternant code $\cI\code^{(\ell)}$ where $\code \in \mathbb{A}_{\va}^d$ and any set $\cE\subset [n]$ of $|\cE|=t$ error positions, where $n\coloneqq |\va|$. A lower bound on the success probability is given in~\cite[Section II.C]{roth2014coding} as
\begin{equation}
  \begin{aligned}
    \Psuc(\cI&\code^{(\ell)}, \cE)\geq 1-\Pr\{\rank(\E)< 2t-d+2\}\\
    & = 1-q^{-(\ell+d-1-2t)(d-1-t)}(1+o(1))\\
    & = 1-q^{-2(t-\frac{3(d-1)+\ell}{4})^2+\frac{(d-1-\ell)^2}{8}}(1+o(1))\ ,
    \label{eq:Pf_large_ell_Roth}
  \end{aligned}
\end{equation}
where $o(1)$ is an expression that goes to $0$ as  $q\to \infty$.

Note that though the decoder in~\cite{roth2014coding} can be applied to interleaved alternant codes, the above lower bound is an asymptotic result. For some applications of alternant codes that we are interested in, e.g., Goppa codes in McEliece system, the field size $q$ is required to be finite or rather small. Therefore, in order to be self-contained and have a general expression on the failure probability, we prove in~\cref{lem:largeEllNoFailCond} that~\cref{algo:SyndromeDecoder} in~\cref{sec:decoding_details} will always succeed in decoding linearly dependent error patterns if $\rank(\E)\geq 2t-d+2$ and we then give a lower bound in~\cref{thm:LargeEll} on the success probability for $\ell\geq t$.

\begin{lemma}\label{lem:largeEllNoFailCond}
  Assume $\ell\geq t$. Let $\mathcal{IC}^{(\ell)}$ be an $\ell$-interleaved alternant code with $\code \in \mathbb{A}^d_{\va}$ and $\cE=\{j_1,j_2,\dots,j_t\}\subset [n]$ be a set of $|\cE|=t$ error positions, where $n\coloneqq |\va|$. For a codeword $\C\in \mathcal{IC}^{(\ell)}$, an error matrix $\widetilde{\E} \in \Fq^{\ell \times n}$ with $\supp(\widetilde{\E}) \coloneqq \cE$ and $\E \coloneqq \widetilde{\E}|_{\cE} \sim \EB{q}{\ell}{t}$, and a received word $\R \coloneqq \C + \widetilde{\E}$, \cref{algo:SyndromeDecoder} \emph{succeeds}, i.e., returns $\hat \C = \C$, if
  \begin{align*}
\rank(\E)\geq 2t-d+2\ .
  \end{align*}
\end{lemma}
\begin{proof}
  Recall from~\cref{lem:FailureCrux} that the decoding does not succeed if and only if
  \begin{equation*}
    \exists \v \in \Fqm^t\setminus\{\0\} \ \text{such that}\ \H \cdot \diag(\v) \cdot \E^\top = \0 \ ,
  \end{equation*}
  where $\H\cdot \diag(\v)$ is the parity check matrix of the code $\GRS_{\va|_{\cE},\v}^{d-t}$.

  We will show that this condition can not be fulfilled \emph{if} $\rank(\E)\geq 2t-d+2$.

  Assume $\rank{\E} \geq 2t-d+2$ and denote $\wt(\v)=w\leq t$. If $w<d-t$, it can readily be seen that no $\E\neq \0$ exists such that $\H \cdot \diag(\v) \cdot \E^\top = \0$ since any $d-t-1$ columns of $\H$ are linearly independent. Now we consider the case $w\geq d-t$. Denote $\bar{\H} = \H|_{\supp(\v)}$, $\bar{\v}=\v|_{\supp(\v)}$, and $\bar{\E}= \E|_{\supp(\v)}$ the respective restrictions to the support of $\v$. Observe the equivalence
\begin{equation}\label{eq:failConditionLargeEll}
  \H|_{\cE} \cdot \diag(\v) \cdot \E^\top = \0 \quad \Leftrightarrow \quad \bar{\H} \cdot \diag(\bar{\v}) \cdot \bar{\E}^\top = \0 \ .
\end{equation}
Note that
\begin{align*}
  \rank(\bar{\E}) &\geq \rank(\E) - (t-w) \\
  &\geq 2t-d+2-(t-w) = w-(d-t)+2 \ .
\end{align*}
and $\bar{\H} \cdot \diag(\bar{\v})$ is the parity check matrix of a $[w,w-(d-t)+1,d-t]_{q^m}$ GRS code.

By definition of the parity check matrix, all rows of $\bar{\E}$ have to be codewords of this GRS code for~\cref{eq:failConditionLargeEll} to be fulfilled. In other words, the code spanned by $\bar{\E}$ needs to be a subcode of this GRS code, i.e., $\myspan{\bar{\E}} \subseteq \myspan{\bar{\H} \cdot \diag(\bar{\v})}^\perp$. This is a contradiction, as by assumption $\dim(\myspan{\bar{\E}}) = \rank(\bar{\E}) \geq w-(d-t)+2$.

\end{proof}

\begin{theorem}\label{thm:LargeEll}
  Assume $\ell\geq t$. Let $\mathcal{IC}^{(\ell)}$ be an $\ell$-interleaved alternant code with $\code \in \mathbb{A}^d_{\va}$ and $\cE=\{j_1,j_2,\dots,j_t\}\subset [n]$ be a set of $|\cE|=t$ error positions, where $n\coloneqq |\va|$. For a codeword $\C\in \mathcal{IC}^{(\ell)}$, an error matrix $\widetilde{\E} \in \Fq^{\ell \times n}$ with $\supp(\widetilde{\E}) \coloneqq \cE$ and $\E \coloneqq \widetilde{\E}|_{\cE} \in \EB{q}{\ell}{t}$, and a received word $\R \coloneqq \C + \widetilde{\E}$, \cref{algo:SyndromeDecoder} \emph{succeeds}, i.e., returns $\hat \C = \C$, with probability
  \begin{align*}
    \Psuc(\cI\code^{(\ell)},\cE)&\geq %
                       \frac{\sum\limits_{s=2t-d+2}^{t}N(\ell,t,s)}{(q^\ell-1)^t}\ ,
  \end{align*}
  where
  \begin{align*}
    N(\ell,t,s)\coloneqq& \, |\{\E \in \EB{q}{\ell}{t} \ | \ \rank(\E)=s\}|\\
    =& \, \sum_{j=0}^{t-s}(-1)^j\binom{t}{j}\prod_{i=0}^{s-1}\frac{(q^\ell-q^i)(q^{t-j}-q^i)}{q^s-q^i}\ .
  \end{align*}
\end{theorem}
\begin{proof}

  By~\cref{lem:largeEllNoFailCond}, it can be readily seen that the success probability is bounded from below by
   \begin{align*}
     \Psuc(\cI\code^{(\ell)},\cE)& \geq \frac{|\{\E\in\EB{q}{\ell}{t} \ | \ \rank(\E)\geq 2t-d+2 \}|}{|\EB{q}{\ell}{t}|}\\
     & = \frac{\sum\limits_{s=2t-d+2}^{t}|\{\E\in\EB{q}{\ell}{t} \ | \ \rank(\E)=s\}|}{(q^\ell-1)^t}\ .
   \end{align*}

   It remains to determine
   \begin{align*}
     N(\ell,t,s)=|\{\E\in\EB{q}{\ell}{t}:\rank(\E)=s\}| \ ,
   \end{align*}
   i.e., the number of matrices of $\F_q^{\ell \times t}$ without any all-zero columns and of a given rank. The number of matrices, including those \emph{with} all-zero columns, of certain rank is given by~\cite{Landsberg1893}\cite[Theorem~2]{FisherAlex1966} %
   \begin{align*}
     M(\ell,t,s) &\coloneqq |\{\E\in\Fq^{\ell\times t} \ | \ \rank(\E)=s\}|\\
                   &=\prod_{i=0}^{s-1}\frac{(q^\ell-q^i)(q^t-q^i)}{q^s-q^i}\ .
   \end{align*}
   To obtain $N(\ell, t,s)$, we need to exclude the matrices with all-zero columns from $M(\ell,t,s)$.
   By the inclusion-exclusion principle, we have
   \begin{align*}
     N(\ell,t,s) &= \sum_{j=0}^{t-s}(-1)^j\binom{t}{j}M(\ell,t-j,s)\ .
   \end{align*}

\end{proof}
\begin{remark}
  This bound is not tight, since even if $\rank(\E)< 2t-d+2$, it is still possible that not every row of $\E$ is in the alternant code $\GRS_{\va|_{\cE},\v'}^{d-t}\cap \Fq$ for any $\v'\in\Fqm^t\setminus\{\0\}$. In other words,~\cref{lem:largeEllNoFailCond} is only a sufficient condition for a successful decoding.
\end{remark}

\section{An Upper Bound on the Probability of Successful Decoding}\label{sec:lowerBound}

For interleaved GRS codes it is known \cite{schmidt2009collaborative} that the probability of a decoding failure, and by \cref{lem:FailRankCond} also the probability of unsuccessful decoding, decreases exponentially in the difference between the number of errors and the maximal decoding radius of \cref{eq:tmaxGRS}. While the numerical results show that this probability is larger for interleaved alternant codes, it nevertheless quickly drops to values out of range for simulation. To evaluate the performance of the lower bounds of \cref{sec:upperBounds}, we derive an upper bound on the probability of a decoding success, by showing that for a certain set of error matrices the decoder given in \cref{algo:SyndromeDecoder} is \emph{never} successful and then analyzing its cardinality.

We begin with a technical statement on the cardinality of the set of these ``bad'' matrices.

\begin{lemma}\label{lem:rankOneSubmatrices}
  Denote by $\Ebbad$ the set of matrices $\E \in \EB{q}{\ell}{t}$ for which there exists a subset of at least $w$ linearly dependent columns. Then
  \begin{align*}
    \max_{w\leq \xi \leq t} \{Z^\xi\} &\leq |\Ebbad| \leq (t-w+1) \max_{w\leq \xi\leq t} \{Z^\xi\} \\
    Z^\xi &\coloneqq \sum_{j=1}^{\floor{\frac{t}{\ell}}} (-1)^{j-1} \binom{\frac{q^\ell-1}{q-1}}{j} D^\xi_{j} \\
    D^\xi_{j} &\coloneqq \left( \prod_{z=0}^{j-1} \binom{t-z\xi}{\xi}\right) (q-1)^{j\xi} (q^\ell - q^{j})^{t-j \xi} \ .
  \end{align*}
\end{lemma}
\begin{proof}
  Consider the equivalence relation $\equiv_q$ on $\Fq^\ell \setminus \{\0\}$ defined by $\v \equiv_{q} \ve{u}$ if  there exists $\lambda \in \Fq^\star$ such that $\v = \lambda \ve{u}$.
  For a fixed vector $\e \in \Fq^\ell\setminus \{\0\}$ and a matrix $\E\in \EB{q}{\ell}{w}$ denote $\delta^\e_{\E} = |\{ i \ | \ E_{:,i} \equiv_q \e\}|$, i.e., the multiplicity of $\e$ among the multiset of columns of $\E$ under the given equivalence relation. For a set of representatives $\cS \subset \Fq^\ell\setminus \{\0\}$ under the given equivalence relation, we have
  \begin{align*}
    D^\xi_{|\cS|} &\coloneqq |\{ \E \in \EB{q}{\ell}{w} \ | \ \delta^{\e}_\E = \xi \  \forall \ \e \in \cS \}| \\
    &= \left( \prod_{z=0}^{|\cS|-1} \binom{t-z\xi}{\xi}\right) (q-1)^{|\cS|\xi} (q^\ell - q^{|\cS|})^{t-|\cS| \xi} \ ,
  \end{align*}
  where the first term accounts for the positions of the vectors of $\cS$ in $\E$, the second term is the number of choices for the scalar coefficients of these positions, and the third term is the number of choices for the remaining columns, namely any non-zero vector that is not equivalent to any element of $\cS$.
  By the principle of inclusion-exclusion we get
  \begin{align*}
    \cZ^\xi &\coloneqq \{ \E \in \EB{q}{\ell}{w} \ | \ \exists \e \in \Fq^\ell \setminus \{\0\} : \delta^{\e}_\E = \xi  \} \\
    Z^\xi &\coloneqq |\cZ^\xi| = \sum_{j=1}^{\floor{\frac{t}{\xi}}} (-1)^{j-1} \binom{\frac{q^{\ell}-1}{q-1}}{j} D^\xi_{j} \ .
  \end{align*}
  The lemma statement follows from the observation that
  \begin{align*}
     \Ebbad = \bigcup_{j=w}^{t} \cZ^{j} \ .
  \end{align*}
\end{proof}

Using the lower bound on the cardinality of $\Ebbad$, we now derive an upper bound on the probability of successful decoding, by showing that the decoder never succeeds if the error matrix is in this set.

\begin{theorem}[Upper Bound on $\Psuc$]\label{thm:lowerBound}
Let $\mathcal{IC}^{(\ell)}$ be an $\ell$-interleaved alternant code with $\code \in \mathbb{A}^d_{\va}$ and $\cE=\{j_1,j_2,\dots,j_t\}\subset [n]$ be a set of $|\cE|=t$ error positions, where $n\coloneqq |\va|$. For a codeword $\C\in \mathcal{IC}^{(\ell)}$, an error matrix $\widetilde{\E} \in \Fq^{\ell \times n}$ with $\supp(\widetilde{\E}) \coloneqq \cE$ and $\E \coloneqq \widetilde{\E}|_{\cE} \sim \EB{q}{\ell}{t}$, and a received word $\R \coloneqq \C + \widetilde{\E}$ \cref{algo:SyndromeDecoder} \emph{succeeds}, i.e., returns $\hat \C = \C$, with probability
\begin{align*}
  \Psuc(\mathcal{IC}^{(\ell)},\cE) &\leq  1-\frac{\max_{d-t \leq \xi \leq t} \{Z^\xi\} }{(q^\ell-1)^t} \ ,
  \end{align*}
  where $Z^\xi$ is given in \cref{lem:rankOneSubmatrices}.
\end{theorem}
\begin{proof}
First observe that each summand in \cref{eq:ieq1} gives a lower bound on the probability of unsuccessful decoding. Therefore, the fraction of matrices $\E \in \EB{q}{\ell}{w}$ that fulfills
\begin{equation}\label{eq:failConditionW}
  \exists \v \in \Fqm^t \ \text{with} \ \wt(\v) = d-t\ \text{such that}\ \H \cdot \diag(\v) \cdot \E^\top = \0 \ ,
\end{equation}
where $\H\in\Fqm^{(d-t-1)\times t}$ denotes the parity-check matrix of the code $\GRS_{\va|_\cE, \boldsymbol{1}}^{d-t}$, gives a lower bound on the probability of unsuccessful decoding $1-\Psuc(\mathcal{IC}^{(\ell)},\cE)$. We denote by $\Ebbad \subset \EB{q}{\ell}{t}$ the set of matrices $\E \in \EB{q}{\ell}{t}$ that fulfills \cref{eq:failConditionW} and show that any error matrix $\E \in \EB{q}{\ell}{t}$ for which there exists a subset $\cL \subset [t]$ of at least $d-t$ columns such that $\rank(\E|_{\cL}) = 1$ fulfills \cref{eq:failConditionW} and is therefore in $\Ebbad$.

Let $\cL \subset [t]$ be a set of size $|\cL| = d-t$ and $\v\in \Fqm^t$ be a vector with $\supp(\v) = \cL$. Denote by $\bar{\H} = \H|_{\cL}$, $\bar{\va} =(\va|_\cE)|_\cL$, $\bar{\v}=\v|_{\cL}$, and $\bar{\E}= \E|_{\cL}$ the respective restrictions to the support $\cL$ of $\v$. Observe the equivalence
\begin{equation}\label{eq:equivalenceSupp}
  \H \cdot \diag(\v) \cdot \E = \0 \quad \Leftrightarrow \quad \bar{\H} \cdot \diag(\bar{\v}) \cdot \bar{\E} = \0 \ .
\end{equation}
Recall that $\E$ has no all-zero columns by definition. As $\bar{\H}\cdot \diag(\bar{\v}) \in \Fqm^{(d-t-1) \times d-t}$ is the parity check matrix of a GRS code, it is of full-rank $d-t-1$ and the dimension of its right kernel is exactly $1$. We conclude that for any $\E \in \EB{q}{\ell}{t}$ that fulfills \cref{eq:equivalenceSupp} there \emph{necessarily} exists a subset $\cL$ of $d-t$ columns such that $\rank(\E|_{\cL}) = 1$.

To show that this is also sufficient, first note that all rows $\bar{\E}_{i,:}$ of this rank $1$ matrix are scalar multiples of some vector $\e\in (\Fq^{\star})^{d-t}$, where at least one scalar is non-zero (recall that $\E$ does not have any all-zero columns). For any fixed $\v \in \Fqm^t$ with $\supp(\v)=\cL$, the matrix $\bar{\H}$ is the parity-check matrix of a $[d-t,1,d-t]$ GRS code, and therefore the $\Fqm$-kernel of $\bar{\H}$ consists of the $\Fqm$-scalar multiples of one vector $\e' \in (\Fq^\star)^{d-t}$. Further, as $\v$ can be any vector of support $\supp(\v)=\cL$, there exists a $\v$ such that $\bar{\H} \cdot \e' = \0$ \emph{for any} $\e' \in (\Fqm^\star)^{d-t}$, and, in particular, for any $\e \in (\Fq^\star)^{d-t}$. It follows that there exists a $\v$ such that \cref{eq:equivalenceSupp} is fulfilled and we conclude that the condition is also \emph{sufficient}.

  A set $\cL \subset [t]$ with $|\cL|=d-t$ such that $\rank(\E|_{\cL})=1$ exists if and only if a subset of $d-t\leq \xi \leq t$ columns in $\E$ are equivalent. Thus, by~\cref{lem:rankOneSubmatrices}, the probability of successful decoding is bounded from above by
  \begin{align*}
    \Psuc(\mathcal{IC}^{(\ell)},\cE) \leq 1- \frac{|\Ebbad|}{|\EB{q}{\ell}{w}|} &\leq 1- \frac{\max_{d-t \leq \xi \leq t} \{Z^\xi\} }{(q^\ell-1)^t} \ .
  \end{align*}
\end{proof}

\section{Discussion and Numerical Results}\label{sec:numericalResults}

\begin{table}
  \centering
  \caption{Overview of the bounds shown in \cref{fig:plots1,fig:plots2}}
  \begin{tabularx}{\linewidth}{llX}%
    \textbf{Label} & \textbf{Defined in} & \textbf{Description} \\ \hline \\
    $\labelRS$ & \cref{thm:boundIRS} & Lower bound on the probability of successful decoding for interleaved RS codes\\[.5em]
    $\labelMain$ & \cref{thm:failureProbNewBound} & Lower bound on the probability of successful decoding for interleaved alternant codes where the minimum of the Singleton, Griesmer, Hamming, Plotkin, Elias, and Linear Programming bound is used for $\kopt$. \\
    $\labelSingleton$ & \cref{thm:failureProbNewBound} & Lower bound on the probability of successful decoding for interleaved alternant codes, where the Singleton bound is used for $\kopt$. \\
    $\labelLz$ & \cref{cor:failureProbNewBoundL0} & Simplified version of \cref{thm:failureProbNewBound}. The minimum of the Singleton, Griesmer, Hamming, Plotkin, Elias, and Linear Programming bound is used for $\kopt$. \\[.5em]
    $\labelLarge$ & \cref{thm:LargeEll} & Lower bound on the probability of successful decoding for interleaved alternant codes with $\ell \geq t$ \\[.5em]
    $\labelMisc$ & \cref{col:miscorrectionAlternant} & Upper bound on the probability of a miscorrection for interleaved alternant codes. We assume that the decoding radius of the interleaved decoder is $\floor{\frac{\ell}{\ell+1}(d-1)}$, i.e., the largest number of errors for which the \emph{RS interleaved decoder}, given in \cref{algo:SyndromeDecoder}, would succeed (see~\cref{rem:alternantMaxRadius}).\\[.5em]
    $\labelLower$ & \cref{thm:lowerBound} & Upper bound on the probability of successful decoding for interleaved alternant codes.\\[.5em]
    $\labelSim$ & \cref{rem:alternantMaxRadius} & Threshold number of errors such that for all numbers of errors left of the indicated line, the interleaved alternant decoder succeeds with a probability of $\Psuc > 0.9$ obtained by simulation with $100$ decoding iterations per parameter set.
  \end{tabularx}
  \label{tab:boundsInPlots}
\end{table}

In \cref{sec:upperBounds,sec:lowerBound} we have established lower and upper bounds on the probability of successful decoding
\begin{align*}
  \Psuc = 1-\Pfail-\Pmisc
\end{align*}
for the interleaved decoding algorithm of \cite{FengTzeng1991ColaDec,schmidt2009collaborative} when applied to interleaved alternant codes for uniformly distributed errors of a given weight. In the following we present and discuss some numerical results, where we compare these upper and lower bounds\footnote{For better presentation, we plot the respective bounds on the probability of \emph{unsuccessful} decoding $1-\Psuc$ instead of the bounds on $\Psuc$.}. In order to better emphasize the individual contributions of failures and miscorrections, we further include an upper bound on the probability of miscorrection $\Pmisc$, given in the Appendix, in the plots of \cref{fig:plots1,fig:plots2}. We label, summarize, and describe the different bounds and versions thereof in \cref{tab:boundsInPlots} and, for convenience and clarity, refer to them by their respective label for the remainder of this section. Further, we fix the code length to be $n = q^m-1$, i.e., given the base field size $q$ and extension degree $m$ we construct the longest possible RS/alternant codes, while excluding $\alpha_i = 0$ as a code locator (see~\cref{def:GRScodes}).

Aside from the comparison of the lower and upper bounds on the success probability, it is also interesting to see how the probability of successful decoding of an interleaved alternant codes compares to that of the corresponding interleaved GRS code over $\Fqm$. Such a bound was derived\footnote{The bound in \cite{schmidt2009collaborative} is presented as a bound on the probability of failure, but it is in fact a bound on the probability of unsuccessful decoding (see~\cref{rem:applicationToRS}).} and shown to be close to probability of successful decoding obtained from simulation in~\cite{schmidt2009collaborative}. For the reader's convenience we restate it in~\cref{thm:boundIRS} and assign it the label $\labelRS$. Note that the decoder employed in~\cite{schmidt2009collaborative} is equivalent to the decoder considered in this work (see~\cref{algo:SyndromeDecoder}), however the error matrix $\widetilde{\E}$ is assumed to be over $\Fqm$ (the field of the RS code) in~\cref{thm:boundIRS}.
\begin{theorem}[Probability of successful decoding for interleaved RS codes {\cite[Theorem~7]{schmidt2009collaborative}}]\label{thm:boundIRS}
  Let $\mathcal{IC}^{(\ell)}$ be an $\ell$-interleaved GRS code with $\GRSp \in \mathbb{G}^d_{\va}$ as in~\cref{def:GRScodes} and $\cE=\{j_1,j_2,\dots,j_t\}\subset [n]$ be a set of $|\cE|=t$ error positions, where $n\coloneqq |\va|$.
  For a codeword $\C\in \mathcal{IC}^{(\ell)}$, an error matrix $\widetilde{\E} \in \Fqm^{\ell \times n}$ with $\supp(\widetilde{\E}) \coloneqq \cE$ and $\E \coloneqq \widetilde{\E}|_{\cE} \sim \EB{q^m}{\ell}{t}$, and a received word $\R \coloneqq \C + \widetilde{\E}$, \cref{algo:SyndromeDecoder} \emph{succeeds}, i.e., returns $\hat \C = \C$, with probability
\begin{align}\label{eq:boundIRS}
  \Psuc(\mathcal{IC}^{(\ell)},\cE)\geq 1-\left(\frac{q^{m\ell}-\frac{1}{q^m}}{q^{m\ell}-1}  \right)^t\cdot \frac{q^{-m(\ell+1)(t_{\max,\mathsf{RS}}-t)}}{q^m-1}\ ,
\end{align}
where $t_{\max,\mathsf{RS}}=\frac{\ell}{\ell+1}(d-1)$. %
\end{theorem}

Before we discuss the numerical evaluations of the bounds, we make an important observation based on the \emph{simulation} results.
\begin{remark}\label{rem:alternantMaxRadius}
  For most parameters the provided lower bounds on the success probability of decoding interleaved alternant codes do not provide a non-trivial bound for the same decoding radius as the bounds for interleaved RS codes of \cite{schmidt2009collaborative}. To determine the real decoding threshold, i.e., the smallest number of errors for which the decoder succeeds with non-negligible probability\footnote{We arbitrarily choose this probability to be $\Psuc > 0.9$ and run $100$ decoding iterations for each parameter set to determine the decoding threshold.}, we rely on simulation results. This threshold is indicated in the plots and labeled $\labelSim$. Notably, for all tested parameters, \emph{the threshold for interleaved alternant codes is the same as for interleaved RS codes}, i.e., the simulation results imply that the collaboratively decoding errors in an interleaved alternant code succeeds w.h.p. for any number of errors $t$ with
  \begin{align*}
    t \leq  \frac{\ell}{\ell+1}(d-1) = t_{\max, \mathsf{RS}} \ .
  \end{align*}
\end{remark}

The numerical evaluations of the bounds are given in \cref{fig:plots1,fig:plots2} for different base field size $q$, extension degree $m$, and distance $d$, each for varying interleaving order $\ell$:
\begin{itemize}
\item $\mathbf{q=2, m=10, d=51}$: The rate of these codes\footnote{Recall that interleaving does not change the rate of the code.} is $R=\frac{k}{n} \approx 0.5$, assuming $k=n-(d-1)m$ (which tends to be true for most alternant codes). For wild Goppa and BCH codes the rate is $R_{\mathsf{Gop./BCH}} \approx 0.75$ (see~\cref{rem:BCHgoppaCodes}). Figs. \ref{subfig:q=2_m=10_r=50_l=50} and \ref{subfig:q=2_m=10_r=50_l=80} are included to show the comparison between $\labelMain$ and $\labelLarge$.
\item $\mathbf{q=2, m=11, d=101}$: For comparison to the parameters stated above, in Figs. \ref{subfig:q=2_m=11_r=100_l=2}, \ref{subfig:q=2_m=11_r=100_l=5}, \ref{subfig:q=2_m=11_r=100_l=10}, and \ref{subfig:q=2_m=11_r=100_l=25} we fix the rate $R=\frac{k}{n} \approx 0.5$ ($R_{\mathsf{Gop./BCH}} \approx 0.75$), increase $m$, and vary $d$ accordingly.
\item $\mathbf{q=32, m=2, d=51}$: To illustrate the influence of the base field size $q$, in Figs. \ref{subfig:q=32_m=2_r=50_l=2} and \ref{subfig:q=32_m=2_r=50_l=10} we show some evaluations for $q=32$.
\end{itemize}

We now briefly discuss the main observations taken from the numerical results. As $\labelSingleton$ and $\labelLz$ are simplifications of $\labelMain$ and therefore strictly worse, we leave their comparison to each other until later in the section, and begin by only comparing $\labelRS, \labelMain, \labelLarge, \labelMisc$, and $\labelLower$. All statements on the decoding failure, miscorrection, and success probability refer to the syndrome-based collaborative decoder of \cite{FengTzeng1991ColaDec,schmidt2009collaborative} given in \cref{algo:SyndromeDecoder}.

\begin{itemize}
  \item For fixed $q,m$, and $\ell$, the probability of a decoding success is significantly lower for interleaved ($q$-ary) alternant codes than for interleaved ($q^m$-ary) RS codes, as even the \emph{upper} bound $\labelLower$ on the success probability for interleaved alternant codes is in most cases smaller than the \emph{lower} bound $\labelRS$ on the success probability for interleaved RS codes.

  \item The probability of unsuccessful decoding interleaved alternant codes $1-\Psuc$ is dominated by the probability of failure $\Pfail$, as $\Pmisc \ll 1-\Psuc$, i.e., the bound on the probability of a miscorrection $\Pmisc$, labeled $\labelMisc$, is multiple orders of magnitude smaller than $1-\Psuc=\Pmisc+\Pfail$ for the best bound on $\Psuc$ among $\labelMain$ and $\labelLarge$. This is consistent with the numerical results from \cite{schmidt2009collaborative} for the case of decoding interleaved RS codes.

  \item For most parameters $\labelMain$ provides the best lower bound on the probability of success $\Psuc$. In particular, for higher interleaving order $\ell$ and relatively small number of errors $t$, it essentially matches the upper bound of \cref{thm:lowerBound} (see~Figs. \ref{subfig:q=2_m=10_r=50_l=10}, \ref{subfig:q=2_m=11_r=100_l=10}, \ref{subfig:q=2_m=10_r=50_l=25}, and \ref{subfig:q=2_m=11_r=100_l=25}).

  \item For fixed $q,m$, and $d$, the relative gap between the number of errors for which the lower bounds on the probability of decoding success become nontrivial, i.e., give $\Psuc >0$, and the simulated decoding threshold decreases for increasing interleaving order~$\ell$ (compare \cref{subfig:q=2_m=10_r=50_l=2}, \ref{subfig:q=2_m=10_r=50_l=5}, \ref{subfig:q=2_m=10_r=50_l=10}, and \ref{subfig:q=2_m=10_r=50_l=25} or \cref{subfig:q=2_m=11_r=100_l=2}, \ref{subfig:q=2_m=11_r=100_l=5}, \ref{subfig:q=2_m=11_r=100_l=10}, and \ref{subfig:q=2_m=11_r=100_l=25} ).

  \item The lower bound $\labelLarge$ on the probability of decoding success for $\ell > t$ improves upon the bound of $\labelMain$ for large interleaving order and number of errors close to the maximum decoding radius (see~\cref{rem:alternantMaxRadius}).
\end{itemize}

Now consider the different versions of the bound in \cref{thm:failureProbNewBound} labeled $\labelMain, \labelSingleton$, and $\labelLz$.
\begin{itemize}
  \item For small $q$, the performance of \cref{thm:failureProbNewBound} is significantly worse when using a field size independent bound for $\kopt$, as evident from comparing $\labelMain$ and $\labelSingleton$ in Figs.~\ref{subfig:q=2_m=10_r=50_l=2} to \ref{subfig:q=2_m=11_r=100_l=25}, \ref{subfig:q=2_m=10_r=50_l=50} and \ref{subfig:q=2_m=10_r=50_l=80}. This can be expected due to the increasing gap between $\kopt$ and the Singleton bound for decreasing $q$.%

  \item For larger interleaving order $\ell$, the simplified lower bound on the probability of successful decoding $\labelLz$ approaches the best version of the bound $\labelMain$ (see Figs.~\ref{subfig:q=2_m=10_r=50_l=10} to~\ref{subfig:q=2_m=11_r=100_l=25}, \ref{subfig:q=2_m=10_r=50_l=50} and \ref{subfig:q=2_m=10_r=50_l=80}).

\end{itemize}

\begin{figure*}
\begin{subfigure}{.5\textwidth}
\centering
\begin{tikzpicture}
\pgfplotsset{compat = 1.3}
\begin{axis}[
	legend style={nodes={scale=0.5, transform shape}},
	width = 0.9\columnwidth,
	height = 0.6\columnwidth,
	xlabel = {{Number of errors $t$}},
	ylabel = {{Probability $1-\Psuc$ or $\Pmisc$}},
	xmin = 25,
	xmax = 34,
	ymin = 1.271205e-09,
	ymax = 10,
	legend pos = south east,
	legend cell align=left,
	ymode=log,
	grid=both]

\addplot [dotted, color=violet, mark=*, mark size=1pt] table[x=t,y=RS] {figs/data/boundsData_q=2_m=10_r=50_l=2.dat};

\addlegendentry{{$1-\labelRS$}};

\addplot [solid, color=blue, mark=triangle, thick, mark size=1pt] table[x=t,y=Thm1] {figs/data/boundsData_q=2_m=10_r=50_l=2.dat};

\addlegendentry{{$1-\labelMain$}};

\addplot [solid, color=cyan, mark=triangle, thick, mark size=1pt] table[x=t,y=WoKopt] {figs/data/boundsData_q=2_m=10_r=50_l=2.dat};

\addlegendentry{{$1-\labelSingleton$}};

\addplot [solid, color=pink, mark=x, thick, mark size=1pt] table[x=t,y=L01] {figs/data/boundsData_q=2_m=10_r=50_l=2.dat};

\addlegendentry{{$1-\labelLz$}};

\addplot [solid, color=lightgray, mark=none, thick, mark size=1pt] table[x=t,y=LowerIE] {figs/data/boundsData_q=2_m=10_r=50_l=2.dat};

\addlegendentry{{$1-\labelLower$}};

\addplot [solid, color=green, mark=none, thick, mark size=1pt] table[x=t,y=Miscorrection] {figs/data/boundsData_q=2_m=10_r=50_l=2.dat};

\addlegendentry{{$\labelMisc$}};

\addplot [solid, color=red, dashed, mark=none, thick, mark size=1pt] table[x=t,y=Sim] {figs/data/boundsData_q=2_m=10_r=50_l=2.dat};

\addlegendentry{{$\labelSim$}};

\end{axis}
\end{tikzpicture}
\caption{$q=2, m=10, d=51, \ell=2$}
\label{subfig:q=2_m=10_r=50_l=2}
\end{subfigure}
\begin{subfigure}{.5\textwidth}
\centering
\begin{tikzpicture}
\pgfplotsset{compat = 1.3}
\begin{axis}[
	legend style={nodes={scale=0.5, transform shape}},
	width = 0.9\columnwidth,
	height = 0.6\columnwidth,
	xlabel = {{Number of errors $t$}},
	ylabel = {{Probability $1-\Psuc$ or $\Pmisc$}},
	xmin = 50,
	xmax = 67,
	ymin = 2.996352e-21,
	ymax = 10,
	legend pos = south east,
	legend cell align=left,
	ymode=log,
	grid=both]

\addplot [dotted, color=violet, mark=*, mark size=1pt] table[x=t,y=RS] {figs/data/boundsData_q=2_m=11_r=100_l=2.dat};

\addlegendentry{{$1-\labelRS$}};

\addplot [solid, color=blue, mark=triangle, thick, mark size=1pt] table[x=t,y=Thm1] {figs/data/boundsData_q=2_m=11_r=100_l=2.dat};

\addlegendentry{{$1-\labelMain$}};

\addplot [solid, color=cyan, mark=triangle, thick, mark size=1pt] table[x=t,y=WoKopt] {figs/data/boundsData_q=2_m=11_r=100_l=2.dat};

\addlegendentry{{$1-\labelSingleton$}};

\addplot [solid, color=pink, mark=x, thick, mark size=1pt] table[x=t,y=L01] {figs/data/boundsData_q=2_m=11_r=100_l=2.dat};

\addlegendentry{{$1-\labelLz$}};

\addplot [solid, color=lightgray, mark=none, thick, mark size=1pt] table[x=t,y=LowerIE] {figs/data/boundsData_q=2_m=11_r=100_l=2.dat};

\addlegendentry{{$1-\labelLower$}};

\addplot [solid, color=green, mark=none, thick, mark size=1pt] table[x=t,y=Miscorrection] {figs/data/boundsData_q=2_m=11_r=100_l=2.dat};

\addlegendentry{{$\labelMisc$}};

\addplot [solid, color=red, dashed, mark=none, thick, mark size=1pt] table[x=t,y=Sim] {figs/data/boundsData_q=2_m=11_r=100_l=2.dat};

\addlegendentry{{$\labelSim$}};

\end{axis}
\end{tikzpicture}
\caption{$q=2, m=11, d=101, \ell=2$}
\label{subfig:q=2_m=11_r=100_l=2}
\end{subfigure}

\begin{subfigure}{.5\textwidth}
\centering
\begin{tikzpicture}
\pgfplotsset{compat = 1.3}
\begin{axis}[
	legend style={nodes={scale=0.5, transform shape}},
	width = 0.9\columnwidth,
	height = 0.6\columnwidth,
	xlabel = {{Number of errors $t$}},
	ylabel = {{Probability $1-\Psuc$ or $\Pmisc$}},
	xmin = 25,
	xmax = 42,
	ymin = 9.389114e-35,
	ymax = 10,
	legend pos = south east,
	legend cell align=left,
	ymode=log,
	grid=both]

\addplot [dotted, color=violet, mark=*, mark size=1pt] table[x=t,y=RS] {figs/data/boundsData_q=2_m=10_r=50_l=5.dat};

\addlegendentry{{$1-\labelRS$}};

\addplot [solid, color=blue, mark=triangle, thick, mark size=1pt] table[x=t,y=Thm1] {figs/data/boundsData_q=2_m=10_r=50_l=5.dat};

\addlegendentry{{$1-\labelMain$}};

\addplot [solid, color=cyan, mark=triangle, thick, mark size=1pt] table[x=t,y=WoKopt] {figs/data/boundsData_q=2_m=10_r=50_l=5.dat};

\addlegendentry{{$1-\labelSingleton$}};

\addplot [solid, color=pink, mark=x, thick, mark size=1pt] table[x=t,y=L01] {figs/data/boundsData_q=2_m=10_r=50_l=5.dat};

\addlegendentry{{$1-\labelLz$}};

\addplot [solid, color=lightgray, mark=none, thick, mark size=1pt] table[x=t,y=LowerIE] {figs/data/boundsData_q=2_m=10_r=50_l=5.dat};

\addlegendentry{{$1-\labelLower$}};

\addplot [solid, color=green, mark=none, thick, mark size=1pt] table[x=t,y=Miscorrection] {figs/data/boundsData_q=2_m=10_r=50_l=5.dat};

\addlegendentry{{$\labelMisc$}};

\addplot [solid, color=red, dashed, mark=none, thick, mark size=1pt] table[x=t,y=Sim] {figs/data/boundsData_q=2_m=10_r=50_l=5.dat};

\addlegendentry{{$\labelSim$}};

\end{axis}
\end{tikzpicture}
\caption{$q=2, m=10, d=51, \ell=5$}
\label{subfig:q=2_m=10_r=50_l=5}
\end{subfigure}%
\begin{subfigure}{.5\textwidth}
\centering
\begin{tikzpicture}
\pgfplotsset{compat = 1.3}
\begin{axis}[
	legend style={nodes={scale=0.5, transform shape}},
	width = 0.9\columnwidth,
	height = 0.6\columnwidth,
	xlabel = {{Number of errors $t$}},
	ylabel = {{Probability $1-\Psuc$ or $\Pmisc$}},
	xmin = 50,
	xmax = 84,
	ymin = 9.677669e-72,
	ymax = 10,
	legend pos = south east,
	legend cell align=left,
	ymode=log,
	grid=both]

\addplot [dotted, color=violet, mark=*, mark size=1pt] table[x=t,y=RS] {figs/data/boundsData_q=2_m=11_r=100_l=5.dat};

\addlegendentry{{$1-\labelRS$}};

\addplot [solid, color=blue, mark=triangle, thick, mark size=1pt] table[x=t,y=Thm1] {figs/data/boundsData_q=2_m=11_r=100_l=5.dat};

\addlegendentry{{$1-\labelMain$}};

\addplot [solid, color=cyan, mark=triangle, thick, mark size=1pt] table[x=t,y=WoKopt] {figs/data/boundsData_q=2_m=11_r=100_l=5.dat};

\addlegendentry{{$1-\labelSingleton$}};

\addplot [solid, color=pink, mark=x, thick, mark size=1pt] table[x=t,y=L01] {figs/data/boundsData_q=2_m=11_r=100_l=5.dat};

\addlegendentry{{$1-\labelLz$}};

\addplot [solid, color=lightgray, mark=none, thick, mark size=1pt] table[x=t,y=LowerIE] {figs/data/boundsData_q=2_m=11_r=100_l=5.dat};

\addlegendentry{{$1-\labelLower$}};

\addplot [solid, color=green, mark=none, thick, mark size=1pt] table[x=t,y=Miscorrection] {figs/data/boundsData_q=2_m=11_r=100_l=5.dat};

\addlegendentry{{$\labelMisc$}};

\addplot [solid, color=red, dashed, mark=none, thick, mark size=1pt] table[x=t,y=Sim] {figs/data/boundsData_q=2_m=11_r=100_l=5.dat};

\addlegendentry{{$\labelSim$}};

\end{axis}
\end{tikzpicture}
\caption{$q=2, m=11, d=101, \ell=5$}
\label{subfig:q=2_m=11_r=100_l=5}
\end{subfigure}

\begin{subfigure}{.5\textwidth}
\centering
\begin{tikzpicture}
\pgfplotsset{compat = 1.3}
\begin{axis}[
	legend style={nodes={scale=0.5, transform shape}},
	width = 0.9\columnwidth,
	height = 0.6\columnwidth,
	xlabel = {{Number of errors $t$}},
	ylabel = {{Probability $1-\Psuc$ or $\Pmisc$}},
	xmin = 25,
	xmax = 46,
	ymin = 1.564516e-71,
	ymax = 10,
	legend pos = south east,
	legend cell align=left,
	ymode=log,
	grid=both]

\addplot [dotted, color=violet, mark=*, mark size=1pt] table[x=t,y=RS] {figs/data/boundsData_q=2_m=10_r=50_l=10.dat};

\addlegendentry{{$1-\labelRS$}};

\addplot [solid, color=blue, mark=triangle, thick, mark size=1pt] table[x=t,y=Thm1] {figs/data/boundsData_q=2_m=10_r=50_l=10.dat};

\addlegendentry{{$1-\labelMain$}};

\addplot [solid, color=cyan, mark=triangle, thick, mark size=1pt] table[x=t,y=WoKopt] {figs/data/boundsData_q=2_m=10_r=50_l=10.dat};

\addlegendentry{{$1-\labelSingleton$}};

\addplot [solid, color=pink, mark=x, thick, mark size=1pt] table[x=t,y=L01] {figs/data/boundsData_q=2_m=10_r=50_l=10.dat};

\addlegendentry{{$1-\labelLz$}};

\addplot [solid, color=lightgray, mark=none, thick, mark size=1pt] table[x=t,y=LowerIE] {figs/data/boundsData_q=2_m=10_r=50_l=10.dat};

\addlegendentry{{$1-\labelLower$}};

\addplot [solid, color=green, mark=none, thick, mark size=1pt] table[x=t,y=Miscorrection] {figs/data/boundsData_q=2_m=10_r=50_l=10.dat};

\addlegendentry{{$\labelMisc$}};

\addplot [solid, color=red, dashed, mark=none, thick, mark size=1pt] table[x=t,y=Sim] {figs/data/boundsData_q=2_m=10_r=50_l=10.dat};

\addlegendentry{{$\labelSim$}};

\end{axis}
\end{tikzpicture}
\caption{$q=2, m=10, d=51, \ell=10$}
\label{subfig:q=2_m=10_r=50_l=10}
\end{subfigure}%
\begin{subfigure}{.5\textwidth}
\centering
\begin{tikzpicture}
\pgfplotsset{compat = 1.3}
\begin{axis}[
	legend style={nodes={scale=0.5, transform shape}},
	width = 0.9\columnwidth,
	height = 0.6\columnwidth,
	xlabel = {{Number of errors $t$}},
	ylabel = {{Probability $1-\Psuc$ or $\Pmisc$}},
	xmin = 50,
	xmax = 91,
	ymin = 1.738559e-146,
	ymax = 10,
	legend pos = south east,
	legend cell align=left,
	ymode=log,
	grid=both]

\addplot [dotted, color=violet, mark=*, mark size=1pt] table[x=t,y=RS] {figs/data/boundsData_q=2_m=11_r=100_l=10.dat};

\addlegendentry{{$1-\labelRS$}};

\addplot [solid, color=blue, mark=triangle, thick, mark size=1pt] table[x=t,y=Thm1] {figs/data/boundsData_q=2_m=11_r=100_l=10.dat};

\addlegendentry{{$1-\labelMain$}};

\addplot [solid, color=cyan, mark=triangle, thick, mark size=1pt] table[x=t,y=WoKopt] {figs/data/boundsData_q=2_m=11_r=100_l=10.dat};

\addlegendentry{{$1-\labelSingleton$}};

\addplot [solid, color=pink, mark=x, thick, mark size=1pt] table[x=t,y=L01] {figs/data/boundsData_q=2_m=11_r=100_l=10.dat};

\addlegendentry{{$1-\labelLz$}};

\addplot [solid, color=lightgray, mark=none, thick, mark size=1pt] table[x=t,y=LowerIE] {figs/data/boundsData_q=2_m=11_r=100_l=10.dat};

\addlegendentry{{$1-\labelLower$}};

\addplot [solid, color=green, mark=none, thick, mark size=1pt] table[x=t,y=Miscorrection] {figs/data/boundsData_q=2_m=11_r=100_l=10.dat};

\addlegendentry{{$\labelMisc$}};

\addplot [solid, color=red, dashed, mark=none, thick, mark size=1pt] table[x=t,y=Sim] {figs/data/boundsData_q=2_m=11_r=100_l=10.dat};

\addlegendentry{{$\labelSim$}};

\end{axis}
\end{tikzpicture}
\caption{$q=2, m=11, d=101, \ell=10$}
\label{subfig:q=2_m=11_r=100_l=10}
\end{subfigure}

\begin{subfigure}{.5\textwidth}
\centering
\begin{tikzpicture}
\pgfplotsset{compat = 1.3}
\begin{axis}[
	legend style={nodes={scale=0.5, transform shape}},
	width = 0.9\columnwidth,
	height = 0.6\columnwidth,
	xlabel = {{Number of errors $t$}},
	ylabel = {{Probability $1-\Psuc$ or $\Pmisc$}},
	xmin = 25,
	xmax = 49,
	ymin = 6.271928e-180,
	ymax = 10,
	legend pos = south east,
	legend cell align=left,
	ymode=log,
	grid=both]

\addplot [dotted, color=violet, mark=*, mark size=1pt] table[x=t,y=RS] {figs/data/boundsData_q=2_m=10_r=50_l=25.dat};

\addlegendentry{{$1-\labelRS$}};

\addplot [solid, color=blue, mark=triangle, thick, mark size=1pt] table[x=t,y=Thm1] {figs/data/boundsData_q=2_m=10_r=50_l=25.dat};

\addlegendentry{{$1-\labelMain$}};

\addplot [solid, color=cyan, mark=triangle, thick, mark size=1pt] table[x=t,y=WoKopt] {figs/data/boundsData_q=2_m=10_r=50_l=25.dat};

\addlegendentry{{$1-\labelSingleton$}};

\addplot [solid, color=pink, mark=x, thick, mark size=1pt] table[x=t,y=L01] {figs/data/boundsData_q=2_m=10_r=50_l=25.dat};

\addlegendentry{{$1-\labelLz$}};

\addplot [solid, color=lightgray, mark=none, thick, mark size=1pt] table[x=t,y=LowerIE] {figs/data/boundsData_q=2_m=10_r=50_l=25.dat};

\addlegendentry{{$1-\labelLower$}};

\addplot [solid, color=green, mark=none, thick, mark size=1pt] table[x=t,y=Miscorrection] {figs/data/boundsData_q=2_m=10_r=50_l=25.dat};

\addlegendentry{{$\labelMisc$}};

\addplot [solid, color=red, dashed, mark=none, thick, mark size=1pt] table[x=t,y=Sim] {figs/data/boundsData_q=2_m=10_r=50_l=25.dat};

\addlegendentry{{$\labelSim$}};

\end{axis}
\end{tikzpicture}
\caption{$q=2, m=10, d=51, \ell=25$}
\label{subfig:q=2_m=10_r=50_l=25}
\end{subfigure}%
\begin{subfigure}{.5\textwidth}
\centering
\begin{tikzpicture}
\pgfplotsset{compat = 1.3}
\begin{axis}[
	legend style={nodes={scale=0.5, transform shape}},
	width = 0.9\columnwidth,
	height = 0.6\columnwidth,
	xlabel = {{Number of errors $t$}},
	ylabel = {{Probability $1-\Psuc$ or $\Pmisc$}},
	xmin = 50,
	xmax = 97,
	ymin = 1.096826e-320,
	ymax = 10,
	legend pos = south east,
	legend cell align=left,
	ymode=log,
	grid=both]

\addplot [dotted, color=violet, mark=*, mark size=1pt] table[x=t,y=RS] {figs/data/boundsData_q=2_m=11_r=100_l=25.dat};

\addlegendentry{{$1-\labelRS$}};

\addplot [solid, color=blue, mark=triangle, thick, mark size=1pt] table[x=t,y=Thm1] {figs/data/boundsData_q=2_m=11_r=100_l=25.dat};

\addlegendentry{{$1-\labelMain$}};

\addplot [solid, color=cyan, mark=triangle, thick, mark size=1pt] table[x=t,y=WoKopt] {figs/data/boundsData_q=2_m=11_r=100_l=25.dat};

\addlegendentry{{$1-\labelSingleton$}};

\addplot [solid, color=pink, mark=x, thick, mark size=1pt] table[x=t,y=L01] {figs/data/boundsData_q=2_m=11_r=100_l=25.dat};

\addlegendentry{{$1-\labelLz$}};

\addplot [solid, color=lightgray, mark=none, thick, mark size=1pt] table[x=t,y=LowerIE] {figs/data/boundsData_q=2_m=11_r=100_l=25.dat};

\addlegendentry{{$1-\labelLower$}};

\addplot [solid, color=green, mark=none, thick, mark size=1pt] table[x=t,y=Miscorrection] {figs/data/boundsData_q=2_m=11_r=100_l=25.dat};

\addlegendentry{{$\labelMisc$}};

\addplot [solid, color=red, dashed, mark=none, thick, mark size=1pt] table[x=t,y=Sim] {figs/data/boundsData_q=2_m=11_r=100_l=25.dat};

\addlegendentry{{$\labelSim$}};

\end{axis}
\end{tikzpicture}
\caption{$q=2, m=11, d=101, \ell=25$}
\label{subfig:q=2_m=11_r=100_l=25}
\end{subfigure}
\caption{Comparison of the bounds for different parameters. For the bounds $\labelRS, \labelMain, \labelSingleton, \labelLz, \labelLarge,$ and $\labelLower$ on the success probability we show the respective probabilities of unsuccessful decoding $1-\Psuc$.}
\label{fig:plots1}
\end{figure*}

\begin{figure*}
\begin{subfigure}{.5\textwidth}
\centering
\begin{tikzpicture}
\pgfplotsset{compat = 1.3}
\begin{axis}[
	legend style={nodes={scale=0.5, transform shape}},
	width = 0.9\columnwidth,
	height = 0.6\columnwidth,
	xlabel = {{Number of errors $t$}},
	ylabel = {{Probability $1-\Psuc$ or $\Pmisc$}},
	xmin = 25,
	xmax = 50,
	ymin = 1.295002e-310,
	ymax = 10,
	legend pos = south east,
	legend cell align=left,
	ymode=log,
	grid=both]

\addplot [dotted, color=violet, mark=*, mark size=1pt] table[x=t,y=RS] {figs/data/boundsData_q=2_m=10_r=50_l=50.dat};

\addlegendentry{{$1-\labelRS$}};

\addplot [solid, color=blue, mark=triangle, thick, mark size=1pt] table[x=t,y=Thm1] {figs/data/boundsData_q=2_m=10_r=50_l=50.dat};

\addlegendentry{{$1-\labelMain$}};

\addplot [solid, color=cyan, mark=triangle, thick, mark size=1pt] table[x=t,y=WoKopt] {figs/data/boundsData_q=2_m=10_r=50_l=50.dat};

\addlegendentry{{$1-\labelSingleton$}};

\addplot [solid, color=pink, mark=x, thick, mark size=1pt] table[x=t,y=L01] {figs/data/boundsData_q=2_m=10_r=50_l=50.dat};

\addlegendentry{{$1-\labelLz$}};

\addplot [solid, color=orange, mark=triangle, thick, mark size=1pt] table[x=t,y=LargeEll] {figs/data/boundsData_q=2_m=10_r=50_l=50.dat};

\addlegendentry{{$1-\labelLarge$}};

\addplot [solid, color=lightgray, mark=none, thick, mark size=1pt] table[x=t,y=LowerIE] {figs/data/boundsData_q=2_m=10_r=50_l=50.dat};

\addlegendentry{{$1-\labelLower$}};

\addplot [solid, color=green, mark=none, thick, mark size=1pt] table[x=t,y=Miscorrection] {figs/data/boundsData_q=2_m=10_r=50_l=50.dat};

\addlegendentry{{$\labelMisc$}};

\addplot [solid, color=red, dashed, mark=none, thick, mark size=1pt] table[x=t,y=Sim] {figs/data/boundsData_q=2_m=10_r=50_l=50.dat};

\addlegendentry{{$\labelSim$}};

\end{axis}
\end{tikzpicture}
\caption{$q=2, m=10, d=51, \ell=50$}
\label{subfig:q=2_m=10_r=50_l=50}
\end{subfigure}%
\begin{subfigure}{.5\textwidth}
\centering
\begin{tikzpicture}
\pgfplotsset{compat = 1.3}
\begin{axis}[
	legend style={nodes={scale=0.5, transform shape}},
	width = 0.9\columnwidth,
	height = 0.6\columnwidth,
	xlabel = {{Number of errors $t$}},
	ylabel = {{Probability $1-\Psuc$ or $\Pmisc$}},
	xmin = 25,
	xmax = 50,
	ymin = 5.188752e-304,
	ymax = 10,
	legend pos = south east,
	legend cell align=left,
	ymode=log,
	grid=both]

\addplot [dotted, color=violet, mark=*, mark size=1pt] table[x=t,y=RS] {figs/data/boundsData_q=2_m=10_r=50_l=80.dat};

\addlegendentry{{$1-\labelRS$}};

\addplot [solid, color=blue, mark=triangle, thick, mark size=1pt] table[x=t,y=Thm1] {figs/data/boundsData_q=2_m=10_r=50_l=80.dat};

\addlegendentry{{$1-\labelMain$}};

\addplot [solid, color=cyan, mark=triangle, thick, mark size=1pt] table[x=t,y=WoKopt] {figs/data/boundsData_q=2_m=10_r=50_l=80.dat};

\addlegendentry{{$1-\labelSingleton$}};

\addplot [solid, color=pink, mark=x, thick, mark size=1pt] table[x=t,y=L01] {figs/data/boundsData_q=2_m=10_r=50_l=80.dat};

\addlegendentry{{$1-\labelLz$}};

\addplot [solid, color=orange, mark=triangle, thick, mark size=1pt] table[x=t,y=LargeEll] {figs/data/boundsData_q=2_m=10_r=50_l=80.dat};

\addlegendentry{{$1-\labelLarge$}};

\addplot [solid, color=lightgray, mark=none, thick, mark size=1pt] table[x=t,y=LowerIE] {figs/data/boundsData_q=2_m=10_r=50_l=80.dat};

\addlegendentry{{$1-\labelLower$}};

\addplot [solid, color=green, mark=none, thick, mark size=1pt] table[x=t,y=Miscorrection] {figs/data/boundsData_q=2_m=10_r=50_l=80.dat};

\addlegendentry{{$\labelMisc$}};

\addplot [solid, color=red, dashed, mark=none, thick, mark size=1pt] table[x=t,y=Sim] {figs/data/boundsData_q=2_m=10_r=50_l=80.dat};

\addlegendentry{{$\labelSim$}};

\end{axis}
\end{tikzpicture}
\caption{$q=2, m=10, d=51, \ell=80$}
\label{subfig:q=2_m=10_r=50_l=80}
\end{subfigure}%

\begin{subfigure}{.5\textwidth}
\centering
\begin{tikzpicture}
\pgfplotsset{compat = 1.3}
\begin{axis}[
	legend style={nodes={scale=0.5, transform shape}},
	width = 0.9\columnwidth,
	height = 0.6\columnwidth,
	xlabel = {{Number of errors $t$}},
	ylabel = {{Probability $1-\Psuc$ or $\Pmisc$}},
	xmin = 25,
	xmax = 34,
	ymin = 6.283091e-34,
	ymax = 10,
	legend pos = south east,
	legend cell align=left,
	ymode=log,
	grid=both]

\addplot [dotted, color=violet, mark=*, mark size=1pt] table[x=t,y=RS] {figs/data/boundsData_q=32_m=2_r=50_l=2.dat};

\addlegendentry{{$1-\labelRS$}};

\addplot [solid, color=blue, mark=triangle, thick, mark size=1pt] table[x=t,y=Thm1] {figs/data/boundsData_q=32_m=2_r=50_l=2.dat};

\addlegendentry{{$1-\labelMain$}};

\addplot [solid, color=cyan, mark=triangle, thick, mark size=1pt] table[x=t,y=WoKopt] {figs/data/boundsData_q=32_m=2_r=50_l=2.dat};

\addlegendentry{{$1-\labelSingleton$}};

\addplot [solid, color=pink, mark=x, thick, mark size=1pt] table[x=t,y=L01] {figs/data/boundsData_q=32_m=2_r=50_l=2.dat};

\addlegendentry{{$1-\labelLz$}};

\addplot [solid, color=lightgray, mark=none, thick, mark size=1pt] table[x=t,y=LowerIE] {figs/data/boundsData_q=32_m=2_r=50_l=2.dat};

\addlegendentry{{$1-\labelLower$}};

\addplot [solid, color=green, mark=none, thick, mark size=1pt] table[x=t,y=Miscorrection] {figs/data/boundsData_q=32_m=2_r=50_l=2.dat};

\addlegendentry{{$\labelMisc$}};

\addplot [solid, color=red, dashed, mark=none, thick, mark size=1pt] table[x=t,y=Sim] {figs/data/boundsData_q=32_m=2_r=50_l=2.dat};

\addlegendentry{{$\labelSim$}};

\end{axis}
\end{tikzpicture}
\caption{$q=32, m=2, d=51, \ell=2$}
\label{subfig:q=32_m=2_r=50_l=2}
\end{subfigure}%
\begin{subfigure}{.5\textwidth}
\centering
\begin{tikzpicture}
\pgfplotsset{compat = 1.3}
\begin{axis}[
	legend style={nodes={scale=0.5, transform shape}},
	width = 0.9\columnwidth,
	height = 0.6\columnwidth,
	xlabel = {{Number of errors $t$}},
	ylabel = {{Probability $1-\Psuc$ or $\Pmisc$}},
	xmin = 25,
	xmax = 46,
	ymin = 6.422853e-323,
	ymax = 10,
	legend pos = south east,
	legend cell align=left,
	ymode=log,
	grid=both]

\addplot [dotted, color=violet, mark=*, mark size=1pt] table[x=t,y=RS] {figs/data/boundsData_q=32_m=2_r=50_l=10.dat};

\addlegendentry{{$1-\labelRS$}};

\addplot [solid, color=blue, mark=triangle, thick, mark size=1pt] table[x=t,y=Thm1] {figs/data/boundsData_q=32_m=2_r=50_l=10.dat};

\addlegendentry{{$1-\labelMain$}};

\addplot [solid, color=cyan, mark=triangle, thick, mark size=1pt] table[x=t,y=WoKopt] {figs/data/boundsData_q=32_m=2_r=50_l=10.dat};

\addlegendentry{{$1-\labelSingleton$}};

\addplot [solid, color=pink, mark=x, thick, mark size=1pt] table[x=t,y=L01] {figs/data/boundsData_q=32_m=2_r=50_l=10.dat};

\addlegendentry{{$1-\labelLz$}};

\addplot [solid, color=lightgray, mark=none, thick, mark size=1pt] table[x=t,y=LowerIE] {figs/data/boundsData_q=32_m=2_r=50_l=10.dat};

\addlegendentry{{$1-\labelLower$}};

\addplot [solid, color=green, mark=none, thick, mark size=1pt] table[x=t,y=Miscorrection] {figs/data/boundsData_q=32_m=2_r=50_l=10.dat};

\addlegendentry{{$\labelMisc$}};

\addplot [solid, color=red, dashed, mark=none, thick, mark size=1pt] table[x=t,y=Sim] {figs/data/boundsData_q=32_m=2_r=50_l=10.dat};

\addlegendentry{{$\labelSim$}};

\end{axis}
\end{tikzpicture}
\caption{$q=32, m=2, d=51, \ell=10$}
\label{subfig:q=32_m=2_r=50_l=10}
\end{subfigure}
\caption{Comparison of the bounds for different parameters. For the bounds $\labelRS, \labelMain, \labelSingleton, \labelLz, \labelLarge,$ and $\labelLower$ on the success probability we show the respective probabilities of unsuccessful decoding $1-\Psuc$.}
\label{fig:plots2}
\end{figure*}

\section{Conclusion and Future Work} \label{sec:conclusion}

In this work, we have presented the first known lower and upper bounds for general parameters on the probability of successfully decoding interleaved alternant codes with the algorithm of \cite{FengTzeng1991ColaDec,schmidt2009collaborative}. The event of a decoding failure was shown to be the main cause of unsuccessful decoding, i.e., miscorrections are negligible in this sense. Numerical evaluations show that one of the provided lower bounds on this probability of successful decoding is tight for some parameters, as it matches the corresponding newly derived upper bound.%

 The most apparent open problem, in particular for smaller interleaving order, is closing the gap between the number of errors for which the bounds provide a nontrivial success probability and the simulated threshold for which the decoder succeeds. A closely related question, which is also of purely theoretical interest, is determining the distribution of the dimensions of all alternant codes for a given set of RS code locators. For specific applications, such as code-based cryptography, improvements of the bounds for other error distributions, arising, e.g., from an additional restriction to full-rank errors, could be of practical relevance. Finally, the simplification of the presented bound on the probability of decoding success, such that an analytical derivation of the maximal number of errors that result in a nontrivial bound is possible, as in the case of interleaved RS codes, is an interesting question to consider.

\bibliographystyle{IEEEtran}
\bibliography{main}

\appendices

\section{Upper Bound on the Miscorrection Probability}

\label{app:miscBound}

We extend the upper bound on the miscorrection probability $\Pmisc$ for interleaved RS codes from~\cite{schmidt2009collaborative} to a bound for interleaved alternant codes. Their strategy applies for any decoder that possess the following property.
\begin{definition}[ML certificate property, {\cite[Definition~3]{schmidt2009collaborative}}]\label{def:MLcertificate}
   Consider a code $\code$ and a received word $\R \coloneqq \C + \widetilde{\E}$ with $\C \in \code$. A decoder of $\code$ is said to have the \emph{ML certificate property} if it always either returns $\hat \C = \arg\min_{\hat \C \in \code} \dcol(\hat \C, \R)$, where $\dcol$ denotes the number of non-zero columns in $\R-\C$, or declares a {\normalfont \texttt{decoding failure}}.
\end{definition}

It was shown in \cite[Theorem~5]{schmidt2009collaborative} that the decoder of \cite{schmidt2009collaborative} for interleaved RS codes has the ML certificate property. As this is a property of the decoder, it clearly also holds when the decoder is applied to any subcode of the interleaved RS codes, in particular for interleaved alternant codes.
However, the bound given in~\cite[Theorem~6]{schmidt2009collaborative} depends on the weight enumerators of the considered code, which are unknown for (interleaved) alternant codes. To circumvent this issue, we slightly generalize~\cite[Theorem~6]{schmidt2009collaborative} by employing general upper bounds on the weight enumerators, thereby making it independent of the specific linear code used.

\begin{theorem}[{\cite{bassalygo1965new,johnson1962new}} (cf.~{\cite[Ch. 17]{macwilliams1977theory}})]\label{thm:JohnsonBound}
  For any code $\code$ of length $n$ and distance $d$ over $\Fq$ it holds that
  \begin{align*}
    A^{\code}_w \leq \frac{\theta_q d n}{w^2-\theta_q n(2w-d)}
  \end{align*}
  with $\theta_q = 1-\frac{1}{q}$, provided the denominator is positive.
\end{theorem}

This bound only applies if the denominator is positive, but we can also find a statement if this is not the case.

\begin{theorem}[{\cite[Theorem~4, Chapter~17]{macwilliams1977theory}}]\label{thm:weightRelation}
  For any code $\code$ of length $n$ and distance $d$ it holds that
  \begin{align*}
    A_w^{\code} \leq \frac{n}{w} \hat A^{[n-1,d]}_{w-1} \ ,
  \end{align*}
  where $\hat A^{[n-1,d]}_{w-1}$ is an upper bound on the $(w-1)$-th weight enumerator of an arbitrary code of length $n-1$ and distance $d$.
\end{theorem}
\begin{proof}
  Note that \cite[Theorem~4, Chapter~17]{macwilliams1977theory} only considers binary codes. However, it is easy to see that it holds for any $q$ by applying the same double counting argument for the number of non-zero positions, instead of the number of ones.
\end{proof}

Finally, we replace the explicit dependence on the weight enumerators in \cite[Theorem~6]{schmidt2009collaborative} by the generic (code independent) bounds of \cref{thm:JohnsonBound,thm:weightRelation}, to obtain an upper bound on the probability of a miscorrection that is valid for any linear code and decoder that exhibits the ML certificate property.

\begin{theorem}[Miscorrection Probability, cf. {\cite[Theorem~6]{schmidt2009collaborative}}]\label{thm:miscorrectionBound}
  Let $\code$ be a linear code of length $n$ and minimum distance $d$ over $\F_{Q}$ decoded with a decoder that exhibits the \emph{ML certificate property} as in \cref{def:MLcertificate}. Assume that the decoding radius of this decoder is $t_{\max}$ and that it decodes a codeword that is corrupted by $t$ errors. Then, the probability of a miscorrection is
  \begin{align*}
    \Pmisc(\code,t) \leq \frac{\sum_{w=d}^{t+\tmax} \hat A_w^{[n,d]} \sum_{\rho=0}^{\min\{t,\tmax\}} U(Q,t,w,\rho)}{\binom{n}{t} (Q-1)^t} \ ,
  \end{align*}
  with
  \begin{align*}
    \hat A^{[n,d]}_w =
    \begin{cases}
      \floor{\frac{\theta_Q d n}{w^2-\theta_Q n(2w-d)}} , & \text{if } w^2>\theta_Q n(2w-d) \\
      \hat A^{[n-1,d]}_{w-1}, & \text{else}
    \end{cases}\ ,
  \end{align*}
  with $\theta_Q = 1-\frac{1}{Q}$, and
  \begin{align}
    U(Q,t,w,\rho)=&\!\sum\limits^{t+w-\rho}_{i=\ceil{\frac{t+w-\rho}{2}}}\!\binom{w}{i}\binom{i}{\rho-(t+w)+2i}\binom{n-w}{t-i} \nonumber\\
     &\qquad \  \cdot(Q-2)^{\rho-(t+w)+2i}(Q-1)^{t-i}\ .\label{eq:Udef}
  \end{align}
  Note that $0^0\coloneqq 1$ so that~\cref{eq:Udef} is valid for the binary case of $Q=2$.
\end{theorem}
\begin{proof}
 Trivially, the bound of \cite[Theorem~6]{schmidt2009collaborative} is increasing in the weight enumerator $A_w$, so replacing them with the upper bound $\hat A^{[n,d]}_w$ obtained from \cref{thm:JohnsonBound,thm:weightRelation} results in a valid upper bound on $\Pmisc$.
\end{proof}

The bound of \cref{thm:miscorrectionBound} is valid for any linear code. For completeness, we explicitly relate its parameters to those of interleaved alternant codes.

\begin{corollary}[Miscorrection Probability of Interleaved Alternant Codes]\label{col:miscorrectionAlternant}
  Let $\mathcal{IC}^{(\ell)}$ be an $\ell$-interleaved alternant code with $\code \in \mathbb{A}^d_{\va}$ and $\cE=\{j_1,j_2,\dots,j_t\}\subset [n]$ be a set of $|\cE|=t$ error positions, where $n\coloneqq |\va|$. For an error matrix $\widetilde{\E} \in \Fq^{\ell \times n}$ with $\supp(\widetilde{\E}) \coloneqq \cE$ and $\widetilde{\E}|_{\cE} \sim \EB{q}{\ell}{t}$ and any decoder with the ML certificate property, the probability of a \emph{miscorrection} for correcting $t\leq \tmax$ errors is upper bounded by \cref{thm:miscorrectionBound} with $Q=q^\ell$.
\end{corollary}
\begin{proof}
  Fix a basis of $\F_{q^\ell}$ over $\F_q$ and regard the code $\mathcal{IC}^{(\ell)}$ as a scalar code over $\F_{q^\ell}$. Clearly, for the minimum distance $\hat d$ of the scalar code it holds that $\hat d = d$ and as~\cref{thm:miscorrectionBound} holds for any linear scalar code, the statement follows.
\end{proof}

We expect \cref{col:miscorrectionAlternant} to be a rather rough upper bound, as it is independent of both, the specific alternant code and its dimension. Nevertheless, it is sufficient for our purpose of showing that the probability of unsuccessful decoding of interleaved alternant codes is dominated by the failure probability, as evident from the numerical results in \cref{sec:numericalResults}.

\section{Interleaved Alternant Codes vs. Interleaved RS Codes}

It is a natural question how interleaved alternant codes compare to interleaved GRS codes of the same cardinality and overall field size. Unfortunately, for the interleaved decoding radius considered in this work, an improvement is only possible in one specific parameter setting, even when considering the improvements in distance or dimension provided by specific class of alternant codes, such as BCH and Goppa codes (see \cref{rem:BCHgoppaCodes}).
\begin{lemma}
Let $m\geq 2$ and consider a $q$-ary alternant code $\cA \in \ALTallp$ of length $n \coloneqq |\va|$, dimension $k_\cA \coloneqq n-(d-1)m$, and distance\footnote{Recall that $d$ is only the designed distance of the alternant code. Specific subclasses of alternant codes are known to have larger distance, see \cref{rem:BCHgoppaCodes}.} $d_\cA\coloneqq \frac{q}{q-1}(d-1)+1$. Let $\code$ be a $q^m$-ary GRS code\footnote{Note that this is \emph{not} the GRS code corresponding to the alternant code $\cA$.} of length $n$ and dimension $k'=k_\cA$. Then, for any $\ell\geq 1$ with $m\mid \ell$ and $\ell'\coloneqq \frac{\ell}{m}$, we have $\dim_{\Fq}(\cI\cA^{(\ell)}) = \dim_{\Fq}(\cI\code^{(\ell')})$ and the decoding radius $\tmax$ of $\cI\cA^{(\ell)}$ exceeds the radius $\tmax'$ of $\cI\code^{(\ell')}$ if and only if $q=m=2$.
\end{lemma}
\begin{IEEEproof}
  The dimensions of $\cI\cA^{(\ell)}$ and $\cI\code^{(\ell')}$ over $\F_q$ follow directly from the definition of $k'$ and $\ell'$, as
  \begin{align*}
\dim_{\Fq}(\cI\cA^{(\ell)}) = k_\cA \ell = k' m \ell' = \dim_{\Fq}(\cI\code^{(\ell')}) \ .
\end{align*}
  The distance of $\cI\code^{(\ell')}$ is $d' = n- k' +1 = (d-1)m+1$. By \cref{eq:tmaxGRS} we have $\tmax > \tmax'$ if and only if
  \begin{align*}
    \frac{\ell}{\ell+1} \Big( \frac{q}{q-1} (d-1) \Big) &> \frac{\ell'}{\ell'+1} \big( (d-1)m \big) \\
    \frac{\ell}{\ell+1} \Big( \frac{q}{q-1} (d-1) \Big) &> \frac{\frac{\ell}{m}}{\frac{\ell}{m}+1} \big( (d-1)m \big) \\
    \frac{1}{\ell+1}  \frac{q}{q-1}   &> \frac{1}{\frac{\ell}{m}+1}  \\
    (\ell+1)  \frac{q-1}{q}   &< \frac{\ell}{m}+1  \\
    \ell \left( \frac{q-1}{q}-\frac{1}{m}\right)   &< 1-\frac{q-1}{q}  \\
    \ell \left( \frac{(m-1)q-m}{qm}\right)   &< \frac{1}{q}  \ .
  \end{align*}
  Now, if $m=q=2$, the left hand side is $0$ and the inequality is fulfilled for any $\ell$. On the other hand, if $m>2$ or $q>2$ we have $(m-1)q-m > 0$ and therefore
  \begin{align*}
    \ell    &< \frac{m}{(m-1)q-m} \leq
              \begin{cases}
                 \frac{m}{m-2}, & \text{\rm if } q=2, \\
                 \frac{m}{2m-3}, & \text{\rm if } q\geq 3 .
              \end{cases}
  \end{align*}
  It is easy to check that for $q=2$ this inequality is fulfilled for $\ell \leq 2$ and $m=3$ or $\ell=1$ and any $m\geq 3$. For $q\geq 3$, the condition is only fulfilled for $\ell = 1$ and $m=2$. However, by definition we have $m \mid \ell$ and this contradiction concludes the proof.
\end{IEEEproof}

While this result shows that interleaved alternant codes generally do not have a larger error correction capability than interleaved RS codes, note that alternant codes have other inherent advantages, as discussed in \cref{sec:introduction}.

\section{Comparison to the $q$-ary Johnson Radius}

It has been shown that any $[n,k,d]_q$ code can be list-decoded up to the $q$-ary Johnson radius \cite{johnson1962new,bassalygo1965new}, i.e., any number of errors $t_\sfJ < \tau_\sfJ$ with
\begin{align*}
  \tau_\sfJ \coloneqq \theta_q n \left( 1-\sqrt{1-\frac{d}{n\theta_q}} \right) \ ,
\end{align*}
where $\theta_q = 1-\frac{1}{q}$, induces a maximal list size that grows polynomially in the code length. For some alternant codes there exist efficient algorithms~\cite{List2011Augot,ListBiGC2013} that allow for decoding up to the binary ($q=2$) Johnson radius. This motivates a comparison between the $q$-ary Johnson radius and the maximal interleaved decoding radius\footnote{Recall that while we are not able to give a theoretical guarantee that the interleaved decoding algorithm succeeds for $\tmax$ errors when applied to interleaved alternant codes, simulation results consistently imply this threshold.} $\tmax$, as given in \cref{eq:tmaxGRS}. First, observe that the radii are upper bounded by $t_\sfJ \leq d-1$ and $\tmax \leq d-2$, respectively. However, for the $q$-ary Johnson radius this maximal value is only achieved for codes where $d$ is close to $\theta_q n$, i.e., that are of very low rate. In contrast, for interleaved alternant codes the proximity of the decoding radius to this upper bound only depends on the interleaving order, which is therefore achievable for codes of any rate. In general, we have $\tmax > t_\sfJ$ for any interleaving order
\begin{align*}
  \ell > \frac{t_\sfJ}{d-1-t_\sfJ} \ ,
\end{align*}
which implies that for any $t_\sfJ < d-2$ there exists an $\ell<d$ such that the decoding radius of $\ell$-interleaved alternant codes is larger than the $q$-ary Johnson radius.
Note that while it is uncommon to compare codes of different overall field size, this comparison is suitable for applications such as code-based cryptography \cite{elleuch2018interleaved,holzbaur2019decoding}, where the limiting factor is the size of the generator matrix.

\end{document}